%% file: main.tex
\documentclass[10pt, fleqn, a4paper]{amsart}

\usepackage{amssymb,amsthm,amsmath}
\usepackage{dcolumn}
\usepackage{bm}
\usepackage[dvipsnames]{xcolor}
\usepackage[linkcolor=Blue,colorlinks=true, colorlinks=true, urlcolor=Blue, citecolor=Blue]{hyperref}

\usepackage{tensor}
\usepackage[italicdiff]{physics}
\usepackage{subcaption}
\usepackage{cancel}

\usepackage{multirow}
\usepackage{tikz-cd}
\usepackage[]{mdframed}
\usepackage{xcolor}
\usepackage{graphicx}

\newtheorem{mythm}{Theorem}[section]

\newtheorem{mylem}[mythm]{Lemma}

\theoremstyle{definition}
\newtheorem{mydef}[mythm]{Definition}

\newtheorem{myrem}[mythm]{Remark}

\newtheorem{myexample}[mythm]{Example}

\numberwithin{equation}{section}

\newmdenv[linecolor = black, frametitle = \colorbox{white}{changeMe}, frametitleaboveskip = -.7em, innertopmargin = -1em, shadow = true, shadowsize = 3pt, skipbelow = 15pt]{equationFrame}
\newmdenv[linecolor = black, frametitle = \colorbox{white}{changeMe}, frametitleaboveskip = -.7em, innertopmargin = -.1em, shadow = true, shadowsize = 3pt, skipbelow = 15pt]{textFrame}

\begin{document}

\title[Killing vector fields and constructing their metric]{Fundamental cosmic anisotropy and its ramifications I: Killing vector fields and constructing their metric}

\author[R.W. Scholtens]{Robbert W.\ Scholtens}
\author[M. Seri]{Marcello Seri}
\author[H. Waalkens]{Holger Waalkens}
\author[R. van de Weygaert]{Rien van de Weygaert}

\address[Robbert W.\ Scholtens, Rien van de Weygaert]{Kapteyn Astronomical Institute \\ University of Groningen, Landleven 12 9747AD Groningen, the Netherlands}
\address[Robbert W.\ Scholtens, Marcello Seri, Holger Waalkens]{Bernoulli Institute for Mathematics, Computer Science, and Artificial Intelligence \\ University of Groningen, Nijenborgh 9 9747AG Groningen, the Netherlands}

\email[Robbert W.\ Scholtens, corresponding author]{r.w.scholtens@rug.nl}

\begin{abstract}
    On the largest scales, the universe appears to be almost homogeneous and isotropic, adhering to the cosmological principle. In contrast, on smaller scales inhomogeneities and anisotropy become increasingly prominent, reflecting the origin, emergence, and formation of structure in the universe. Moreover, a range of tensions between various cosmological observations may suggest it necessary to explore the consequences of departure from the ideal, uniform universe on the fundamental level. Thus, in this work, the foundation of spatially homogeneous yet anisotropic universes is studied. Specifically, when given a 3D Lie algebra of \emph{desired} Killing vector fields (as would be the case for a homogeneous yet anisotropic universe), we provide an explicit construction for the metric that has exactly those as its Killing vector fields. This construction is presented accessibly, in a directly-usable, algorithmic fashion. Some examples demonstrating the construction are worked out, including a constructive method to separate out (cosmic) time dependence in spatially homogeneous, anisotropic cosmologies.
\end{abstract}

\maketitle

\input{introduction}

\input{results}

\input{corollaries}

\input{methodology}

\input{conclusion}

\input{references}

\end{document}

%% file: introduction.tex
\section{Introduction}

The current cosmic worldview is underpinned by the \emph{cosmological principle} (CP), stating that at the largest scales the universe is spatially homogeneous and isotropic. \cite{Peebles:1971, ellis_relativistic_2012}. These two properties impose a strong constraint on the large-scale geometry, leading to the three families of Robertson-Walker (RW) spacetime metrics, and ensuing Friedmann-Lemaître-Robertson-Walker (FLRW) cosmological models. The single permitted freedom of the scale factor $a(t)$ allows linking made observations to cosmic history, e.g\ through the redshift. Though validity of the cosmological principle has been questioned since its postulation (see \cite[\S2]{Peebles:1993} for a discussion), recent observations (to be elaborated in the next subsection) seem to challenge the CP's requirement of isotropy. It is therefore necessary to understand, especially from a theoretical perspective, what structure remains in a homogeneous yet anisotropic universe.

A spacetime which is spatially homogeneous and completely anisotropic is called a \emph{Bianchi model} \cite[\S5.2]{ellis_cosmological_2008}. They have been studied since their introduction into astrophysics in the 1960's \cite{ellis_class_1969, ellis_bianchi_2006}. A result of this endeavor has been to tabulate, for prototypical Lie algebras of vector fields, frames for spacetime metrics on which these vector fields are in fact Killing (see e.g.\ \cite[\S6.4]{ryan_homogeneous_1975}). As homogeneity (and isotropy) are encoded in a metric's Lie algebra of Killing vector fields (KLA), this constitutes a relation between KLAs and metric frames.

The standard methodology suffers two drawbacks, however: (i) the method of derivation is reliant on the abstract theories of Lie groups and algebras, which is not friendly to the uninitiated; and (ii) the tabulations only feature specific realizations of the (desired) KLAs, whereas for applications it may be desirable to have a more general ``recipe.''

Our work remedies these drawbacks. Namely, we present in Section \S\ref{sec:derivation} a methodology to find, from a 3D Lie algebra of \emph{desired} Killing vector fields, a frame for a metric on which it is indeed realized as Killing. Differently from previous literature, we do this without reliance on abstract mathematical theory, enabling cosmologists without the necessary background to benefit from it. Moreover, we present an accessible and step-by-step version of our method in Section \S\ref{sec:table}. In contrast with the tabulations of previous literature, our method may take as input \emph{any} suitable Lie algebra of vector fields.

\subsection*{Challenges to the Cosmological Principle}
While an overwhelming body of observational evidence has shown that the cosmological principle represents an amazingly accurate description of the global universe, at least asymptotically, we still need to establish how close the universe fulfils these conditions and what the scale of homogeneity actually is. First and foremost evidence is that of the amazing isotropy of the Cosmic Microwave Background, at a level of $\approx 10^{-5}$ \cite{cobe1992,wmap2003,planck2020}. Strong indications also exist for the universe to be close to homogeneous on large scales. Most notably this is indicated by the scaling of the angular clustering of galaxies as a function of depth of the galaxy survey: the
angular two-point correlation function of galaxies is seen to scale almost perfectly with the survey depth as expected for a spatial
galaxy distribution that assumes uniformity on large scales \cite{groth1977,connolly2002}. Nonetheless, it has not yet
been unequivocally established beyond what scale the universe may be considered to be homogeneous and isotropic. General contention
is that the universe may be considered homogeneous and isotropic beyond a scale of a few to several hundreds of
Megaparsec, see e.g.\ \cite{Jones:2004}.

An increasing number of observations indicate that although the cosmological principle may be a very accurate description of
physical reality, it may also be in need of either a more precise formulation or, at worst, a reconsideration. Even within the
context of the observed universe, 
it is more accurate to consider the cosmological principle as a statistical
principle. Indeed, since the COBE CMB maps, we know that even on the largest scales the mass distribution in the universe is marked by
a stochastic spatial pattern of density variations (be it of a minute amplitude). In fact, we know that structure in the
universe arose through gravitational growth of these primordial perturbations, and that these were imprinted on all scales.
Since COBE revealed the map of the large scale angular structure of the cosmic microwave background temperature, we know
that at the time of recombination there were structures much larger than the particle horizon at the time. By implication,
this is still true at the current cosmic epoch, so that the visible universe presents us with a minute deviation from
the global cosmological average. Hence, it is more precise to state that the statistical properties of the corresponding
inhomogeneous mass distribution obey the cosmological principle and do not vary with direction and location, but that the
particular realization within our visible universe may not be perfectly uniform, see e.g.\ \cite{Peebles:1993}.

At a more fundamental level, various studies forwarded observational indications for deviations from pure isotropy or
homogeneity, and hence may have serious implications for the validity of FLRW-based cosmologies. An intriguing indication
for possible intrinsic anomalous anisotropies is the finding by Copi et al.\ \cite{copi2004,copi2015} that the
low multipole components of the cosmic microwave background anisotropies are mutually aligned, and aligned with our motion through
the universe (i.e.\ the dipole motion). The existence of this \textit{Axis of Evil} would seriously question the validity of
the isotropy of our universe, and hence the cosmological principle. At an even more critical level, Sarkar and collaborators have
highlighted some serious
problems in the regular interpretation of cosmological observations within the context of the standard FLRW cosmology (see \cite{sarkar_heart_2022} and references therein). One aspect
concerns the fact that the dipole in the CMB is fully ascribed to the gravitational attraction by local inhomogeneities, and hence
can be interpreted as a dipole effect due to our Galaxy moving within the restframe blackbody cavity of the cosmic microwave background
radiation. However, recently Secrest et al.\ \cite{secrest_test_2021} found that the dipole with respect to the high
redshift sample of quasars does not coincide with the CMB dipole, which it would be expected to do when our Galaxy's motion is solely
a local effect. It may imply the existence of an anomalously large dipole that is not of local origin and directly calls
into question the validity of the assumed isotropy of the universe. This finding is augmented by a potentially even
more disruptive result that questions the reality of dark energy and instead may indicate a bulk flow on a cosmological
scale. Colin et al.\ \cite{colin_evidence_2019} inferred from a sample of 740 Supernovae Ia that the cosmological
acceleration parameter is direction dependent, which may imply that we are embedded in a deep bulk flow that conflicts with the
assumption of an isotropic universe.

\subsection*{Bianchi models as anisotropic alternatives}
Following the indications for the reality of possible deviations from the cosmological principle, the purpose of the present study is to explore cosmologies that do not entail uniformity. With the ultimate goal to extend this to the geometric structure of fully inhomogeneous and anisotropic cosmologies, we here first explore the implications for cosmologies that are anisotropic, yet homogeneous. No longer demanding the symmetry of isotropy, there is a substantial extension of possible cosmic metrics with respect to the 3 types of highly symmetric RW metrics. They are generally called \emph{Bianchi models.}


In this work, we investigate the mathematical basis of the Bianchi models from the perspective of the infinitesimal isometries, or \emph{Killing vector fields} (KVFs), belonging to the metrics involved. Namely, requirements on homogeneity and isotropy of a metric manifest themselves mathematically as such infinitesimal isometries. Denoting by $g$ the metric under consideration, a KVF $\xi$ is characterized by the following equations:
\begin{equation}\label{eq:intro-killing}
    \mathcal{L}{_\xi}g=0,\qor\tensor{g}{_\mu_\nu_{,\alpha}}\tensor{\xi}{^\alpha}+2\tensor{g}{_\alpha_{(\mu}}\tensor{\xi}{^\alpha_{,\nu)}}=0,\qor\tensor{\xi}{_{(\mu;\nu)}}=0.
\end{equation}
where $\mathcal{L}$ indicates Lie derivative (also known as Lie drag), comma notation indicates differentiation, and semicolon notation indicates covariant differentiation. For instance, if $g$ is taken as a Robertson-Walker metric,
one will find exactly six solutions for $\xi$ to this equation. Moreover, due to the property of Lie derivatives that $\mathcal{L}_{[\xi,\eta]}=[\mathcal{L}_\xi,\mathcal{L}_\eta]$ --- so that the commutator of KVFs is again a KVF --- we see that KVFs form a Lie algebra of vector fields (Killing Lie algebra; KLA). It is known that a metric in dimension $n$ may allow for up to $n(n+1)/2$ KVFs \cite[\S17.1]{ellis_relativistic_2012}, and certainly does not need to reach this bound.
\begin{figure}[t]
    \centering
    \begin{tikzpicture}[x=0.75pt,y=0.75pt,yscale=-.75,xscale=.75]

        \draw  [dash pattern={on 4.5pt off 4.5pt}] (124,356) .. controls (124,342.19) and (135.19,331) .. (149,331) .. controls (162.81,331) and (174,342.19) .. (174,356) .. controls (174,369.81) and (162.81,381) .. (149,381) .. controls (135.19,381) and (124,369.81) .. (124,356) -- cycle ;
        \draw  [dash pattern={on 4.5pt off 4.5pt}] (86.5,356) .. controls (86.5,321.48) and (114.48,293.5) .. (149,293.5) .. controls (183.52,293.5) and (211.5,321.48) .. (211.5,356) .. controls (211.5,390.52) and (183.52,418.5) .. (149,418.5) .. controls (114.48,418.5) and (86.5,390.52) .. (86.5,356) -- cycle ;
        \draw  [dash pattern={on 4.5pt off 4.5pt}] (49,356) .. controls (49,300.77) and (93.77,256) .. (149,256) .. controls (204.23,256) and (249,300.77) .. (249,356) .. controls (249,411.23) and (204.23,456) .. (149,456) .. controls (93.77,456) and (49,411.23) .. (49,356) -- cycle ;
        \draw  [dash pattern={on 4.5pt off 4.5pt}] (544,356) .. controls (544,342.19) and (555.19,331) .. (569,331) .. controls (582.81,331) and (594,342.19) .. (594,356) .. controls (594,369.81) and (582.81,381) .. (569,381) .. controls (555.19,381) and (544,369.81) .. (544,356) -- cycle ;
        \draw  [dash pattern={on 4.5pt off 4.5pt}] (506.5,356) .. controls (506.5,321.48) and (534.48,293.5) .. (569,293.5) .. controls (603.52,293.5) and (631.5,321.48) .. (631.5,356) .. controls (631.5,390.52) and (603.52,418.5) .. (569,418.5) .. controls (534.48,418.5) and (506.5,390.52) .. (506.5,356) -- cycle ;
        \draw  [dash pattern={on 4.5pt off 4.5pt}] (469,356) .. controls (469,300.77) and (513.77,256) .. (569,256) .. controls (624.23,256) and (669,300.77) .. (669,356) .. controls (669,411.23) and (624.23,456) .. (569,456) .. controls (513.77,456) and (469,411.23) .. (469,356) -- cycle ;
        \draw  [dash pattern={on 4.5pt off 4.5pt}] (322.83,134.65) .. controls (322.83,120.69) and (338.5,109.38) .. (357.83,109.38) .. controls (377.16,109.38) and (392.83,120.69) .. (392.83,134.65) .. controls (392.83,148.61) and (377.16,159.93) .. (357.83,159.93) .. controls (338.5,159.93) and (322.83,148.61) .. (322.83,134.65) -- cycle ;
        \draw  [dash pattern={on 4.5pt off 4.5pt}] (304.88,89.18) .. controls (325.11,65.63) and (365.2,66.91) .. (394.45,92.03) .. controls (423.69,117.14) and (431,156.59) .. (410.78,180.13) .. controls (390.56,203.68) and (350.46,202.4) .. (321.22,177.28) .. controls (291.98,152.16) and (284.66,112.72) .. (304.88,89.18) -- cycle ;
        \draw  [dash pattern={on 4.5pt off 4.5pt}] (273.15,202.16) .. controls (246.24,168.4) and (262.86,110.39) .. (310.27,72.6) .. controls (357.69,34.8) and (417.94,31.53) .. (444.85,65.29) .. controls (471.76,99.05) and (455.14,157.06) .. (407.73,194.85) .. controls (360.31,232.65) and (300.06,235.92) .. (273.15,202.16) -- cycle ;
        \draw [line width=1.5]    (149,356) -- (566,356) ;
        \draw [shift={(569,356)}, rotate = 180] [color={rgb, 255:red, 0; green, 0; blue, 0 }  ][line width=1.5]    (28.42,-8.55) .. controls (18.07,-3.63) and (8.6,-0.78) .. (0,0) .. controls (8.6,0.78) and (18.07,3.63) .. (28.42,8.55)   ;
        \draw [line width=1.5]    (149,356) -- (356.94,135.9) ;
        \draw [shift={(359,133.72)}, rotate = 133.37] [color={rgb, 255:red, 0; green, 0; blue, 0 }  ][line width=1.5]    (28.42,-8.55) .. controls (18.07,-3.63) and (8.6,-0.78) .. (0,0) .. controls (8.6,0.78) and (18.07,3.63) .. (28.42,8.55)   ;
        \draw  [fill={rgb, 255:red, 0; green, 0; blue, 0 }  ,fill opacity=1 ] (144,356) .. controls (144,353.24) and (146.24,351) .. (149,351) .. controls (151.76,351) and (154,353.24) .. (154,356) .. controls (154,358.76) and (151.76,361) .. (149,361) .. controls (146.24,361) and (144,358.76) .. (144,356) -- cycle ;
        \draw  [fill={rgb, 255:red, 0; green, 0; blue, 0 }  ,fill opacity=1 ] (352.83,134.65) .. controls (352.83,131.89) and (355.07,129.65) .. (357.83,129.65) .. controls (360.59,129.65) and (362.83,131.89) .. (362.83,134.65) .. controls (362.83,137.42) and (360.59,139.65) .. (357.83,139.65) .. controls (355.07,139.65) and (352.83,137.42) .. (352.83,134.65) -- cycle ;
        \draw  [fill={rgb, 255:red, 0; green, 0; blue, 0 }  ,fill opacity=1 ] (564,356) .. controls (564,353.24) and (566.24,351) .. (569,351) .. controls (571.76,351) and (574,353.24) .. (574,356) .. controls (574,358.76) and (571.76,361) .. (569,361) .. controls (566.24,361) and (564,358.76) .. (564,356) -- cycle ;
        \draw    (149,356) -- (149.88,484.45) ;
        \draw    (357.83,134.65) -- (358.71,263.1) ;
        \draw    (569,356) -- (569.88,484.45) ;

        \draw (359,359.4) node [anchor=north] [inner sep=0.75pt]  [font=\large]  {$\xi $};
        \draw (252,241.46) node [anchor=south east] [inner sep=0.75pt]  [font=\large]  {$\eta $};
        \draw    (90,54) -- (90,85) -- (56,85)  ;
        \draw (59,58.4) node [anchor=north west][inner sep=0.75pt]  [font=\large]  {$\mathbb{R}^{2}$};
        \draw  [draw opacity=0][fill={rgb, 255:red, 255; green, 255; blue, 255 }  ,fill opacity=1 ]  (141,316) -- (160,316) -- (160,341) -- (141,341) -- cycle  ;
        \draw (144,320.4) node [anchor=north west][inner sep=0.75pt]    {$1$};
        \draw  [draw opacity=0][fill={rgb, 255:red, 255; green, 255; blue, 255 }  ,fill opacity=1 ]  (141,282) -- (160,282) -- (160,307) -- (141,307) -- cycle  ;
        \draw (144,286.4) node [anchor=north west][inner sep=0.75pt]    {$2$};
        \draw  [draw opacity=0][fill={rgb, 255:red, 255; green, 255; blue, 255 }  ,fill opacity=1 ]  (141,243) -- (160,243) -- (160,268) -- (141,268) -- cycle  ;
        \draw (144,247.4) node [anchor=north west][inner sep=0.75pt]    {$3$};
        \draw  [draw opacity=0][fill={rgb, 255:red, 255; green, 255; blue, 255 }  ,fill opacity=1 ]  (560,316) -- (579,316) -- (579,341) -- (560,341) -- cycle  ;
        \draw (563,320.4) node [anchor=north west][inner sep=0.75pt]    {$1$};
        \draw  [draw opacity=0][fill={rgb, 255:red, 255; green, 255; blue, 255 }  ,fill opacity=1 ]  (560,282) -- (579,282) -- (579,307) -- (560,307) -- cycle  ;
        \draw (563,286.4) node [anchor=north west][inner sep=0.75pt]    {$2$};
        \draw  [draw opacity=0][fill={rgb, 255:red, 255; green, 255; blue, 255 }  ,fill opacity=1 ]  (560,243) -- (579,243) -- (579,268) -- (560,268) -- cycle  ;
        \draw (563,247.4) node [anchor=north west][inner sep=0.75pt]    {$3$};
        \draw  [draw opacity=0][fill={rgb, 255:red, 255; green, 255; blue, 255 }  ,fill opacity=1 ]  (375,112) -- (394,112) -- (394,137) -- (375,137) -- cycle  ;
        \draw (378,116.4) node [anchor=north west][inner sep=0.75pt]    {$1$};
        \draw  [draw opacity=0][fill={rgb, 255:red, 255; green, 255; blue, 255 }  ,fill opacity=1 ]  (399,95) -- (418,95) -- (418,120) -- (399,120) -- cycle  ;
        \draw (402,99.4) node [anchor=north west][inner sep=0.75pt]    {$2$};
        \draw  [draw opacity=0][fill={rgb, 255:red, 255; green, 255; blue, 255 }  ,fill opacity=1 ]  (448,86) -- (467,86) -- (467,111) -- (448,111) -- cycle  ;
        \draw (451,90.4) node [anchor=north west][inner sep=0.75pt]    {$3$};
        \draw (149.88,487.85) node [anchor=north] [inner sep=0.75pt]  [font=\large]  {$p$};
        \draw (358.71,266.5) node [anchor=north] [inner sep=0.75pt]  [font=\large]  {$q'$};
        \draw (569.88,487.85) node [anchor=north] [inner sep=0.75pt]  [font=\large]  {$q$};
    \end{tikzpicture}
    \caption{Illustration of effect of Lie dragging, in two dimensions. The dashed loops around $p$, $q$, and $q'$ are level sets of the distance. Dragging the metric from $p$ to $q$ via $\xi$ leaves the level sets intact, so $\xi$ represents an (infinitesimal) isometry of the metric. In contrast, when we drag from $p$ to $q'$ via $\eta$, the level sets are perturbed. This shows the metric has been altered, and hence that $\eta$ is not an (infinitesimal) isometry.}
    \label{fig:LieDrag}
\end{figure}

We will present the mathematical definitions in Sections \S\ref{sec:foundations} and \S\ref{sec:classification}. Intuitively,
\begin{itemize}
    \item \emph{homogeneity} means we can, everywhere, Lie drag with a KVF in any direction, so the metric can be preserved in going in any direction; and
    \item \emph{isotropy} means we have full rotational symmetry, which is defined as KVFs having a nonempty set of fixed points.
\end{itemize}
A Bianchi model is thus a spacetime model wherein we retain homogeneity, yet we have no KVFs which generate a rotational symmetry. It is so named since homogeneity and full anisotropy implies that the KLA has dimension 3, and Lie algebras of dimension 3 were first classified by Bianchi \cite{bianchi1898}.


An example metric of the so-called Bianchi I type is given by
\begin{equation}\label{eq:bianchi-I}
    \dd{s}^2=-\dd{t}^2+a(t)^2\dd{x}^2+b(t)^2\dd{y}^2+c(t)^2\dd{z}^2,
\end{equation}
for $a(t)$, $b(t)$, and $c(t)$ mutually different --- a generalization of the flat RW metric, allowing for different length scales in the different directions. This metric is homogeneous (the translation vector fields $\partial_x$ etc.\ are Killing) yet fully anisotropic, as the ``rotation generating'' vector fields, e.g.\ $x\partial_y-y\partial_x$, fail to be Killing for all $t$. (One could make the KVFs depend on time, but since we wish our symmetries to be purely spatial, this is outside of the scope of this work.)

As a compromise between isotropy and anisotropy, one might consider the metric as in \eqref{eq:bianchi-I} with the restriction $a(t)=b(t)$. This metric then is homogeneous and yet not fully isotropic nor anisotropic: the vector field $x\partial_y-y\partial_x$ is Killing, but $z\partial_x-x\partial_z$ and $y\partial_z-z\partial_y$, the other two rotation-generating vector fields, are not. What we have defined is an example of a \emph{locally rotationally symmetric} (LRS) Bianchi model: a homogeneous model that allows one axis of rotational symmetry. Such models (and in particular LRS Bianchi I models like we exemplified above) have received attention in the literature because, although they are freer than flat FLRW, they are still restricted enough to allow for ease of calculation and interpretation; see, e.g. \cite[\S5.2,\S5.3]{ellis_cosmological_2008}, \cite{elst_covariant_1996}, and \cite{stewart_solutions_1968}. Since in this study we do not wish to consider LRS Bianchi models, we shall discuss it no further.

Bianchi models have, since their introduction in the 1960's \cite{ellis_class_1969}, been investigated extensively in the literature (e.g.\ \cite{collins_why_1973, collins_singularities_1979}) --- see \cite{ellis_bianchi_2006} for a review. A particular draw is that since Bianchi models are still homogeneous, one can make a clear separation between the geometric description, which describes the spatial hypersurfaces at constant time, and the temporal evolution of quantities defined thereon. This then transforms the Einstein equation from a set of (a priori coupled) partial differential equations to one of coupled ordinary differential equations. In this way, it is possible to view model universes as dynamical systems, with relevant parameters (e.g.\ critical density) evolving in cosmic time. The work of Wainwright and collaborators, e.g.\ \cite{wainwright_dynamical_1989, nilsson_dynamical_2000, wainwright_dynamical_2005} is heavily focused on this.

The hope is that, due to their inherent anisotropy, Bianchi models could (partially) explain some of the anisotropies that are currently observed, such as that of the CMB. Particular focus in the literature has been placed on Bianchi types I, V, VII, and IX, as these have properties which makes them likely candidates from a physical perspective (such as isotropization towards late time, or spatial curvature). A calculation detailing contributions of Bianchi-type perturbations to the CMB was performed in \cite{pontzen_rogues_2009}. A Bianchi type VII contribution was tested against the \textit{Planck} data, though it was ruled out \cite{ade_planck_2014}.

\subsection*{Our approach}

A Bianchi model is defined by the KLA of its infinitesimal isometries, which can be computed when given a metric. In order to understand this relation even closer, we ask the reverse question: \emph{given a suitable 3D Lie algebra of vector fields, (how) can we construct an explicit description of a metric for which the given Lie algebra is the KLA, or a subalgebra thereof?} The premise of this question stands in contrast with the usual approach, where at the start a metric frame is chosen so that it is of a specific Bianchi type. Such frames are tabulated, for instance in \cite[\S6.4]{ryan_homogeneous_1975} and \cite[\S8.2]{stephani_exact_2003}.

An answer to our question can be found in the work of Jantzen \cite{jantzen_dynamical_1979}. Here, the author's methodology is to leverage the abstract theory of Lie groups and Lie algebras in order to derive the appropriate metric frames. These frames are then reported in terms of conveniently chosen canonical coordinates, so that the frames are at their simplest. These are the frames that can then also be found in the tabulations referenced above.

Besides the abstract nature of e.g.\ Jantzen's argument (requiring ancillary knowledge of these topics to follow completely), the fact that the metric is presented in a canonical choice of coordinates is suboptimal for some applications. For instance, in situations where a coordinate system has already been imposed (e.g.\ cosmological simulations or experimental apparatus), it may prove difficult and/or inconvenient to change to the ``canonical coordinates'' as found in the existing literature. In such cases, it would be ideal to have a method to find the metric frame \emph{without having to make a specific coordinate choice.} That is, one that works regardless of the coordinatization of the involved desired KVFs.

This has been the aim of the present study, and its fruit is presented in Table \ref{tab:forms} and Section \S\ref{sec:table}. Moreover, we strive to do so in an accessible and pedagogical manner. In this spirit, we only use advanced mathematics to a strictly necessary degree so as to aid this development. For the sake of simplicity and direct applicability, we do not consider Lie algebras of dimensions 4 or higher, and we assume that all isometries are generated by infinitesimal isometries.

\subsection*{Structure of this paper}
In Section \S\ref{sec:table} we present the methodology in answer to our central question. That is, how to go from a set of desired KVFs to a metric frame on which they are indeed realized as KVFs. Subsequently, in Section \S\ref{sec:corollaries}, we apply our method to find the metric frames for some given desired KVFs. Additionally, we show how to \emph{constructively} separate out time-dependence in a given spatially homogeneous metric.

In Section \S\ref{sec:derivation}, we detail our methodology to obtain the result from Section \S\ref{sec:table}. We do this by first, in Sections \S\ref{sec:foundations} and \S\ref{sec:classification}, treat the mathematics pertaining generally to homogeneity, isotropy, and Bianchi models. Then, in Section \S\ref{sec:findingFrame}, we derive concretely the equation(s) to solve, and then proceed to do so via the method of characteristics. After some additional results in Section \S\ref{sec:auxillaries} on uniqueness and naturality, we end the paper with a conclusion and outlook to future research.

\subsection*{Notation}
Throughout the paper we assume the Einstein summation index when an expression contains an upper and lower index. Moreover, as much as possible, indices from the start of the Roman/Greek alphabet are dummy indices, whilst those from mid-alphabet are free. This indexing convention allows for quick distinction between the two. Roman indices are assumed to run $1,2,3$, whilst Greek run $0,1,2,3$.

%% file: results.tex
\section{From KVFs to metric frame}\label{sec:table}

In this section we discuss how to answer our central question: when given a desired set of KVFs, how to find a metric on which they are realized as such. We shall do so by first discussing a bare minimum of preliminaries, and then detailing the steps to follow in order to find a metric \emph{frame.} Any metric with time-dependent components, written in that frame, will then satisfy the query.

Indicate by $\Sigma^3$ the (spatial) three-dimensional hypersurface which we regard as space, and we assume to be simply connected. This will be necessary to ensure the uniqueness of this construction -- see \S\ref{sec:foundations} for further details. For our construction, we need the desired KVFs, $\{\tensor{\xi}{_I}\}_{I=1}^3$ to satisfy some properties.
\begin{enumerate}
    \item They form a 3D Lie algebra, which will be (a sub-algebra of) the Lie algebra of KVFs (KLA) for our to-be-found metric.
    \item At each point in $\Sigma^3$, they are linearly independent; the $\{\tensor{\xi}{_I}\}$ form a frame.
\end{enumerate}
Let us call such a set of desired KVFs a \emph{pre-KLA.}

Our construction will proceed with finding an appropriate \emph{frame} for the metric, which is dubbed the \emph{invariant frame} (IF). In said frame, then, if the metric components depend only on time, we can guarantee that our pre-KLA will indeed be the KLA (or a subalgebra) thereof. Symbolically, as also found in \eqref{eq:basisInsight}:
\begin{equation}
    \{\tensor{X}{_i}\}:\text{\eqref{eq:commutationProperty} holds}\implies \{\tensor{\xi}{_I}\}\text{ are Killing on }g=\tensor{g}{_\mu_\nu}(t)\,\tensor{e}{^\mu}\tensor{e}{^\nu},
\end{equation}
where \eqref{eq:commutationProperty} is $[\tensor{\xi}{_I},\tensor{X}{_i}]=0$, for $\{\tensor{\xi}{_I}\}$ a pre-KLA, and $\{\tensor{e}{^\mu}\}_{\mu=0}^3$ is the dual of $\{\tensor{X}{_\mu}\}$ (set $\tensor{e}{^0}=\dd{t}\iff\tensor{X}{_0}=\pdv*{t}$). In Section \S\ref{sec:derivation} we present our derivation of this problem statement.

The long and short of it, as also forms the conclusion of Section \S\ref{sec:auxillaries}, is that we can algorithmize the solution to $[\tensor{\xi}{_I},\tensor{X}{_i}]=0$, and hence to finding the appropriate metric frame $\{\tensor{e}{^\mu}\}$, by means of the following construction.
\newpage
\begin{textFrame}[frametitle = \colorbox{white}{From pre-KLA to metric frame}]
    Let $\{\tensor{\xi}{_I}\}$ be a pre-KLA (defined above).
    \begin{enumerate}
        \item Perform a $\operatorname{GL}(3)$-transformation $R$ on $\{\tensor{\xi}{_I}\}$, so that its structure constants match one of the items in Table \ref{tab:forms}. (If the structure constants already match, set $R:=\operatorname{id}$.) Dub $\tensor{{\xi'}}{_I}:=R\tensor{\xi}{_I}$.
        \item For $\{\tensor{{\xi'}}{_I}\}$, solve the displayed differential equations for $p(\vb*{x})$, $q(\vb*{x})$, and $r(\vb*{x})$, with initial conditions $p(0)=0$ etc.\
        \item Use $\{\tensor{{\xi'}}{_I}\}$ and the found $p$, $q$, and $r$ to fill in the given relations. We obtain the vector fields $\tensor{{X'}}{_i}$.
        \item Calculate $\tensor{X}{_i}=R^{-1}\tensor{{X'}}{_i}$ with $R$ from the first step.
        \item Find the duals to $\{\tensor{X}{_i}\}$, i.e.\ those 1-forms $\{\tensor{e}{^i}\}$ such that $\langle\tensor{e}{^i},\tensor{X}{_j}\rangle=\tensor*{\delta}{^i_j}$. (Here $\langle\cdot,\cdot\rangle$ indicates vector-covector contraction.)
        \item Set $\tensor{e}{^0}=\dd{t}$, where $t$ is interpreted as time.
    \end{enumerate}
    Then, any spacetime metric
    \begin{equation}\label{eq:separatedFormRecipe}
        g=\tensor{g}{_\mu_\nu}(t)\,\tensor{e}{^\mu}\tensor{e}{^\nu},
    \end{equation}
    where $\tensor{g}{_\mu_\nu}(t)$ are free functions of time (taking into account symmetry and Lorentzian signature) will have $\{\tensor{\xi}{_I}\}$ as (a subset of) its KVFs.
\end{textFrame}

\begin{table*}[h!]
    \centering
    \scriptsize
    \renewcommand{\arraystretch}{1.5}
    \begin{tabular}{|l|cc|}
        \hline
        \textbf{Bianchi I:} all vanish & & \\ \cline{1-1} \multicolumn{3}{|c|}{$\tensor{X}{_i}=\tensor{\xi}{_i}$} \\ \hline\hline
        \textbf{Bianchi II:} $\tensor*{C}{^1_2_3}=1$ & & \\ \cline{1-1} \multicolumn{2}{|c}{$\begin{cases}\tensor{X}{_1}=\tensor{\xi}{_1} \\ \tensor{X}{_2}=r\tensor{\xi}{_1}+\tensor{\xi}{_2} \\ \tensor{X}{_3}=-q\tensor{\xi}{_1}+\tensor{\xi}{_3} \end{cases}$} & $q:\begin{cases}\tensor{\xi}{_1}q=0 \\ \tensor{\xi}{_2}q=1 \\ \tensor{\xi}{_3}q=0\end{cases} r:\begin{cases}\tensor{\xi}{_1}r=0 \\ \tensor{\xi}{_2}r=0 \\ \tensor{\xi}{_3}r=1\end{cases}$ \\ \hline\hline
        \textbf{Bianchi III:} $\tensor*{C}{_1_3^1}=1$ & & \\ \cline{1-1} \multicolumn{2}{|c}{$\begin{cases}\tensor{X}{_1}=e^r\tensor{\xi}{_1} \\ \tensor{X}{_2}=\tensor{\xi}{_2} \\ \tensor{X}{_3}=-p\tensor{\xi}{_1}+\tensor{\xi}{_3} \end{cases}$} & $p:\begin{cases}\tensor{\xi}{_1}p=1 \\ \tensor{\xi}{_2}p=0 \\ \tensor{\xi}{_3}p=p\end{cases} r:\begin{cases}\tensor{\xi}{_1}r=0 \\ \tensor{\xi}{_2}r=0 \\ \tensor{\xi}{_3}r=1\end{cases}$ \\ \hline\hline
        \textbf{Bianchi IV:} $\tensor*{C}{_1_3^1}=\tensor*{C}{_2_3^1}=\tensor*{C}{_2_3^2}=1$ & & \\ \cline{1-1} \multicolumn{2}{|c}{$\begin{cases}\tensor{X}{_1}=e^r\tensor{\xi}{_1} \\ \tensor{X}{_2}=re^r\tensor{\xi}{_1}+e^r\tensor{\xi}{_2} \\ \tensor{X}{_3}=-(p+q)\tensor{\xi}{_1}-q\tensor{\xi}{_2}+\tensor{\xi}{_3}\end{cases}$} & $p:\begin{cases} \tensor{\xi}{_1}(p)=1 \\ \tensor{\xi}{_2}(p)=0 \\ \tensor{\xi}{_3}(p)=p+q\end{cases} q:\begin{cases} \tensor{\xi}{_1}q=0 \\ \tensor{\xi}{_2}q=1 \\ \tensor{\xi}{_3}q=q\end{cases} r:\begin{cases}\tensor{\xi}{_1}r=0 \\ \tensor{\xi}{_2}r=0 \\ \tensor{\xi}{_3}r=e^r\end{cases}$ \\ \hline\hline
        \textbf{Bianchi V:} $\tensor*{C}{_1_3^1}=\tensor*{C}{_2_3^2}=1$ & & \\ \cline{1-1} \multicolumn{2}{|c}{$\begin{cases} \tensor{X}{_1}=e^r\tensor{\xi}{_1} \\ \tensor{X}{_2}=e^r\tensor{\xi}{_2} \\ \tensor{X}{_3}=-p\tensor{\xi}{_1}-q\tensor{\xi}{_2}+\tensor{\xi}{_3} \end{cases}$} & $p:\begin{cases} \tensor{\xi}{_1}p=1 \\ \tensor{\xi}{_2}p=0 \\ \tensor{\xi}{_3}p=p \end{cases} q:\begin{cases}\tensor{\xi}{_1}q=0 \\ \tensor{\xi}{_2}q=1 \\ \tensor{\xi}{_3}q=q\end{cases} r:\begin{cases}\tensor{\xi}{_1}r=0 \\ \tensor{\xi}{_2}r=0 \\ \tensor{\xi}{_3}r=1\end{cases}$ \\ \hline\hline
        \textbf{Bianchi VI$_h$:} $\tensor*{C}{_1_3^1}=1$ and $\tensor*{C}{_2_3^2}=h$ ($h\neq0,1$) & & \\ \cline{1-1} \multicolumn{2}{|c}{$\begin{cases} \tensor{X}{_1}=e^r\tensor{\xi}{_1} \\ \tensor{X}{_2}=e^{hr}\tensor{\xi}{_2} \\ \tensor{X}{_3}=-p\tensor{\xi}{_1}-q\tensor{\xi}{_2}+\tensor{\xi}{_3} \end{cases}$} & $p:\begin{cases} \tensor{\xi}{_1}p=1 \\ \tensor{\xi}{_2}p=0 \\ \tensor{\xi}{_3}p=p \end{cases} q:\begin{cases} \tensor{\xi}{_1}q=0 \\ \tensor{\xi}{_2}q=h \\ \tensor{\xi}{_3}q=hq \end{cases} r:\begin{cases} \tensor{\xi}{_1}r=0 \\ \tensor{\xi}{_2}r=0 \\ \tensor{\xi}{_3}r=1\end{cases}$ \\ \hline\hline
        \textbf{Bianchi VII$_h$:} $\tensor*{C}{_1_3^2}=\tensor*{C}{_3_2^1}=1$ and $\tensor*{C}{_2_3^2}=h$ ($h^2<4$) & & \\ \cline{1-1} \multicolumn{2}{|c}{$\begin{cases} \tensor{X}{_1}=e^{hr/2}\left(\cos(\Delta r)-\tfrac{1}{2}h\Delta^{-1}\sin(\Delta r)\right)\tensor{\xi}{_1} \\ \qquad+e^{hr/2}\Delta^{-1}\sin(\Delta r)\tensor{\xi}{_2} \\ \tensor{X}{_2}=-e^{hr/2}\Delta^{-1}\sin(\Delta r)\tensor{\xi}{_1} \\ \qquad+e^{hr/2}\left(\cos(\Delta r)+\tfrac{1}{2}h\Delta^{-1}\sin(\Delta r)\right)\tensor{\xi}{_2} \\ \tensor{X}{_3}=q\tensor{\xi}{_1}-(p+q)\tensor{\xi}{_2}+\tensor{\xi}{_3} \end{cases}$} & $\begin{aligned}&p:\begin{cases} \tensor{\xi}{_1}p=1 \\ \tensor{\xi}{_2}p=0 \\ \tensor{\xi}{_3}p=q \end{cases} q:\begin{cases} \tensor{\xi}{_1}q=0 \\ \tensor{\xi}{_2}q=-1 \\ \tensor{\xi}{_3}q=p+hq \end{cases} \\ &r:\begin{cases} \tensor{\xi}{_1}r=0 \\ \tensor{\xi}{_2}r=0 \\ \tensor{\xi}{_3}r=-\Delta^2(\tfrac{1}{4}h^2+\Delta^2)^{-1} \end{cases}\end{aligned}$ \\ \hline\hline
        \textbf{Bianchi VIII:} $\tensor*{C}{_3_2^1}=\tensor*{C}{_3_1^2}=\tensor*{C}{_1_2^3}=1$ & & \\ \cline{1-1} \multicolumn{2}{|c}{$\begin{cases}\tensor{X}{_1}=\cosh(q)\cosh(r)\tensor{\xi}{_1}\\ \qquad+(\cosh(r)\sin(p)\sinh(q)-\cos(p)\sinh(r))\tensor{\xi}{_2} \\ \qquad+(\cos(p)\cosh(r)\sinh(q)+\sin(p)\sinh(r))\tensor{\xi}{_3} \\ \tensor{X}{_2}=-\cosh(q)\sinh(r)\tensor{\xi}{_1} \\ \qquad+(\cos(p)\cosh(r)-\sin(p)\sinh(q)\sinh(r))\tensor{\xi}{_2} \\ \qquad-(\cosh(r)\sin(p)+\cos(p)\sinh(q)\sinh(r))\tensor{\xi}{_3} \\ \tensor{X}{_3}=\sinh(q)\tensor{\xi}{_1}+\cosh(q)\sin(p)\tensor{\xi}{_2}+\cos(p)\cosh(q)\tensor{\xi}{_3}\end{cases}$} & $\begin{aligned}&p:\begin{cases}\tensor{\xi}{_1}p=1 \\ \tensor{\xi}{_2}p=-\tanh(q)\sin(p) \\ \tensor{\xi}{_3}p=-\tanh(q)\cos(p) \end{cases} q:\begin{cases}\tensor{\xi}{_1}q=0 \\ \tensor{\xi}{_2}q=\cos(p) \\ \tensor{\xi}{_3}q=-\sin(p) \end{cases} \\ & r:\begin{cases}\tensor{\xi}{_1}r=0 \\ \tensor{\xi}{_2}r=\sin(p)/\cosh(q) \\ \tensor{\xi}{_3}r=\cos(p)/\cosh(q)\end{cases}\end{aligned}$ \\ \hline\hline
        \textbf{Bianchi IX:} $\tensor*{C}{_2_3^1}=\tensor*{C}{_3_1^2}=\tensor*{C}{_1_2^3}=1$ & & \\ \cline{1-1} \multicolumn{2}{|c}{$\begin{cases} \tensor{X}{_1}=\cos(q)\cos(r)\tensor{\xi}{_1} \\ \qquad+\left(\cos(r)\sin(p)\sin(q)-\cos(p)\sin(r)\right)\tensor{\xi}{_2} \\ \qquad+\left(\cos(p)\cos(r)\sin(q)+\sin(p)\sin(r)\right)\tensor{\xi}{_3} \\ \tensor{X}{_2}=\cos(q)\sin(r)\tensor{\xi}{_1}\\ \qquad+\left(\cos(p)\cos(r)+\sin(p)\sin(q)\sin(r)\right)\tensor{\xi}{_2} \\ \qquad+\left(\cos(p)\sin(q)\sin(r)-\cos(r)\sin(p)\right)\tensor{\xi}{_3} \\ \tensor{X}{_3}=-\sin(q)\tensor{\xi}{_1}+\cos(q)\sin(p)\tensor{\xi}{_2}+\cos(q)\cos(p)\tensor{\xi}{_3} \end{cases}$} & $\begin{aligned}&p:\begin{cases}\tensor{\xi}{_1}p=1 \\ \tensor{\xi}{_2}p=\tan(q)\sin(p) \\ \tensor{\xi}{_3}p=\tan(q)\cos(p)\end{cases} q:\begin{cases}\tensor{\xi}{_1}q=0 \\ \tensor{\xi}{_2}q=\cos(p) \\ \tensor{\xi}{_3}q=-\sin(p)\end{cases} \\ & r:\begin{cases}\tensor{\xi}{_1}r=0 \\ \tensor{\xi}{_2}r=\sin(p)/\cos(q) \\ \tensor{\xi}{_3}r=-\cos(p)/\cos(q)\end{cases}\end{aligned}$ \\ \hline
    \end{tabular}
    \caption{General forms of the IFs $\{\tensor{X}{_i}\}$ belonging to pre-KLAs $\{\tensor{\xi}{_a}\}$ of the various Bianchi types. The functions $p=p(\vb*{x})$, $q=q(\vb*{x})$, and $r=r(\vb*{x})$ must be solved for any particular choice of pre-KLA, to satisfy the indicated relations. In Bianchi VII$_h$, $\Delta:=\tfrac{1}{2}\sqrt{4-h^2}$. See Section \S\ref{sec:table} in the text.}
    \label{tab:forms}
    \vfill
\end{table*}

Some remarks on this procedure are in order.
\begin{itemize}
    \item Strictly speaking, step (iv) is optional to perform; one can also simply ``forget'' about the prime and move on with the rest of the construction.
    \item In step (vi), there is considerable freedom in the choice of $t$. This corresponds to a ``choice of time axis,'' and will be touched upon on in Section \S\ref{sec:findingFrame}.
\end{itemize}

In the next section, we will provide some examples of this construction in action, as well as utilizing this to actively bring an appropriate metric (i.e.\ a metric that has a pre-KLA as a subalgebra of its KLA) into the form of \eqref{eq:separatedFormRecipe}.

%% file: corollaries.tex
\section{Applying our result}\label{sec:corollaries}

\subsection{Example I --- Bianchi III}
Suppose we wish to find a spacetime metric so that the following vector fields, adapted from \cite{osetrin_exact_2023}, are Killing:
\begin{equation}
    \tensor{\xi}{_1}=\tensor{\partial}{_2},\quad\tensor{\xi}{_2}=\tensor{\partial}{_1}+\tensor{\partial}{_3}+\tensor{x}{^2}\tensor{\partial}{_2},\qand\tensor{\xi}{_3}=\tensor{\partial}{_3}+\tensor{x}{^2}\tensor{\partial}{_2}.
\end{equation}
We calculate the commutators as $[\tensor{\xi}{_1},\tensor{\xi}{_2}]=\tensor{\xi}{_1}$, $[\tensor{\xi}{_2},\tensor{\xi}{_3}]=0$, and $[\tensor{\xi}{_3},\tensor{\xi}{_1}]=-\tensor{\xi}{_1}$; the structure constants are thus $\tensor{C}{^1_1_2}=\tensor{C}{^1_1_3}=1$ and remainder vanishing. This is of the Bianchi III type: the class of algebras that are not nilpotent but whose derived algebra is 1-dimensional.

According to step (i) we need to perform a $\operatorname{GL}(3)$ transformation transforming the structure constants to $\tensor{C}{^1_1_3}=1$ and remainder vanishing. This is accomplished by
\begin{equation}
    R=\pmqty{1 & 0 & 0 \\ 0 & 1 & -1 \\ 0 & 0 & 1}\iff\begin{cases}\tensor{{\xi'}}{_1}=\tensor{\xi}{_1} \\ \tensor{{\xi'}}{_2}=\tensor{\xi}{_2}-\tensor{\xi}{_3} \\ \tensor{{\xi'}}{_3}=\tensor{\xi}{_3}\end{cases}.
\end{equation}
Then we solve the differential equations for $p$ and $r$ ($q$ does not feature for Bianchi III) using the \emph{primed} $\xi$. We compute that $p(\vb*{x})=\tensor{x}{^2}$ and $r(\vb*{x})=\tensor{x}{^3}$, so upon substitution we find
\begin{equation}
    \tensor{{X'}}{_1}=\exp(\tensor{x}{^3})\tensor{{\xi'}}{_1},\quad\tensor{{X'}}{_2}=\tensor{{\xi'}}{_2},\qand\tensor{{X'}}{_3}=\tensor{{\xi'}}{_3}-\tensor{x}{^2}\tensor{{\xi'}}{_1}.
\end{equation}
We now perform the inverse of $R$ in order to obtain the unprimed $X$, and also substitute in the coordinate description of $\{\tensor{\xi}{_I}\}$:
\begin{equation}
    \tensor{X}{_1}=\exp(\tensor{x}{^3})\tensor{\partial}{_2},\quad\tensor{X}{_2}=\tensor{\partial}{_1}+\tensor{\partial}{_3},\qand\tensor{X}{_3}=\tensor{\partial}{_3}.
\end{equation}
The duals are then found to be
\begin{equation}
    \tensor{e}{^1}=\exp(-\tensor{x}{^3})\dd{\tensor{x}{^2}},\quad\tensor{e}{^2}=\dd{\tensor{x}{^1}},\qand\tensor{e}{^3}=\dd{\tensor{x}{^3}}-\dd{\tensor{x}{^1}}
\end{equation}
Finally, set $\tensor{e}{^0}=\dd{t}$. Then, any spacetime metric of the form
\begin{equation}
    g=\tensor{g}{_\mu_\nu}(t)\,\tensor{e}{^\mu}\tensor{e}{^\nu}
\end{equation}
will have $\tensor{\xi}{_1}$, $\tensor{\xi}{_2}$, and $\tensor{\xi}{_3}$ among its KVFs. For the particular choice $\tensor{g}{_\mu_\nu}=\operatorname{diag}(-1,a(t)^2,b(t)^2,c(t)^2)$, for instance, this metric takes the shape
\begin{equation}
    g=-\dd{t}^2+a^2\exp(-2\tensor{x}{^3})(\dd{\tensor{x}{^2}})^2+b^2(\dd{\tensor{x}{^1}})^2+c^2(\dd{\tensor{x}{^1}}-\dd{\tensor{x}{^3}})^2.
\end{equation}

\subsection{Example II --- twice Bianchi II}
Consider the following two Lie algebras:
\begin{equation}
    \begin{cases} \tensor{\xi}{_1}=\tensor{\partial}{_2} \\ \tensor{\xi}{_2}=\tensor{\partial}{_3} \\ \tensor{\xi}{_3}=\tensor{\partial}{_1}+\tensor{x}{^3}\tensor{\partial}{_2}\end{cases}\qand\quad\begin{cases} \tensor{\eta}{_1}=\tensor{\partial}{_1} \\ \tensor{\eta}{_2}=\tensor{\partial}{_2}-\tfrac{1}{2}\tensor{x}{^3}\tensor{\partial}{_1} \\ \tensor{\eta}{_3}=\tensor{\partial}{_3}+\tfrac{1}{2}\tensor{x}{^2}\tensor{\partial}{_1}\end{cases}.
\end{equation}
The first of these is type II as found in \cite[\S6.4]{ryan_homogeneous_1975}, and the second is the Heisenberg algebra. Both have the tabulated structure constants $\tensor{C}{^1_2_3}=1$ belonging to type II. As both algebras are written in the $\{\tensor{\partial}{_i}\}$-basis we can find a linear transformation mapping one to the other: $\phi:\tensor{\xi}{_I}\mapsto\tensor{\eta}{_I}$, $I=1,2,3$. In the $\{\tensor{\partial}{_i}\}$-basis, it is expressed as
\begin{equation}
    \phi=\pmqty{\tfrac{1}{2}\tensor{x}{^2}-\tensor{x}{^3} & 1 & -\tfrac{1}{2}\tensor{x}{^3} \\ 0 & 0 & 1 \\ 1 & 0 & 0}.
\end{equation}
We note that this is not a Jacobian matrix; hence, there does not exist a diffeomorphism that connects the Lie algebras, despite both being of the same type. This shows, as observed already in~\cite{jantzen_dynamical_1979}, that tables that list only the first type, such as our source, are incomplete in the sense that an inequivalent model of the same type also exists.

Following the steps of Section \S\ref{sec:table} for both Lie algebras, we can find that the IFs for both are given by
\begin{equation}
    \begin{cases} \tensor{X}{_1}=\tensor{\partial}{_2} \\ \tensor{X}{_2}=\tensor{x}{^1}\tensor{\partial}{_2}+\tensor{\partial}{_3} \\ \tensor{X}{_3}=\tensor{\partial}{_1} \end{cases}\qand\quad\begin{cases}\tensor{Y}{_1}=\tensor{\partial}{_1} \\ \tensor{Y}{_2}=\tensor{\partial}{_2}+\tfrac{1}{2}\tensor{x}{^3}\tensor{\partial}{_1} \\ \tensor{Y}{_3}=\tensor{\partial}{_3}-\tfrac{1}{2}\tensor{x}{^2}\tensor{\partial}{_1}\end{cases},
\end{equation}
where the $X$s correspond to the $\xi$s, and the $Y$s to the $\eta$s.\footnote{In fact, if $\phi$ were the Jacobian matrix of a diffeomorphism, then due to the naturality as discussed in Section \S\ref{sec:auxillaries} it would have to map $\phi:\tensor{X}{_i}\mapsto\tensor{Y}{_i}$, $i=1,2,3$; we can verify by direct calculation that this does not occur.} The dual frames follow:
\begin{equation}
    \begin{cases}\tensor{e}{^1}=\dd{\tensor{x}{^2}}-\tensor{x}{^1}\dd{\tensor{x}{^3}} \\ \tensor{e}{^2}=\dd{\tensor{x}{^3}} \\ \tensor{e}{^3}=\dd{\tensor{x}{^1}}\end{cases}\qand\begin{cases} \tensor{f}{^1}=\dd{\tensor{x}{^1}}+\tfrac{1}{2}\left(\tensor{x}{^2}\dd{\tensor{x}{^3}}-\tensor{x}{^3}\dd{\tensor{x}{^2}}\right) \\ \tensor{f}{^2}=\dd{\tensor{x}{^2}} \\ \tensor{f}{^3}=\dd{\tensor{x}{^3}}\end{cases}.
\end{equation}

Assuming a diagonal spatial metric $\tensor{\gamma}{_i_j}=\operatorname{diag}(a(t)^2,b(t)^2,c(t)^2)$ in the $\{\tensor{e}{^i}\}$- and $\{\tensor{f}{^i}\}$-bases, the spatial parts of the metrics when written out in coordinates become
\begin{subequations}
    \begin{equation}\label{eq:stb2-1}
        \begin{split}
            \dd{\ell_\text{RS}^2}&=\tensor{\gamma}{_i_j}\tensor{e}{^i}\tensor{e}{^j} \\
            &=a^2\left(\dd{\tensor{x}{^1}}\right)^2+b^2\left(\dd{\tensor{x}{^3}}\right)^2+c^2\left(\dd{\tensor{x}{^2}}-\tensor{x}{^1}\dd{\tensor{x}{^3}}\right)^2
        \end{split}
    \end{equation}
    and
    \begin{equation}\label{eq:stb2-2}
        \begin{split}
            \dd{\ell_\text{H}^2}&=\tensor{\gamma}{_i_j}\tensor{f}{^i}\tensor{f}{^j} \\
            &=a^2\left(\dd{\tensor{x}{^1}}+\tfrac{1}{2}\left(\tensor{x}{^2}\dd{\tensor{x}{^3}}-\tensor{x}{^3}\dd{\tensor{x}{^2}}\right)\right)^2+b^2\left(\dd{\tensor{x}{^2}}\right)^2+c^2\left(\dd{\tensor{x}{^3}}\right)^2.
        \end{split}
    \end{equation}
\end{subequations}
One can then verify that, regardless of the coefficients, $\{\tensor{\xi}{_I}\}$ are Killing on \eqref{eq:stb2-1} and $\{\tensor{\eta}{_I}\}$ are Killing on \eqref{eq:stb2-2}.

\subsection{Separating the metric}
The existence of the IF when given a pre-KLA not only allows us to formulate metrics on which the pre-KLA is realized as a KLA, but also to rewrite given homogeneous metrics into a standard form, wherein the metric components depend only on time. This is evidently useful for GR-related purposes, as then the Einstein equation will reduce to coupled ODEs, instead of PDEs, simplifying the solving procedure. According to \cite{scholtens_inpub_2025}, we have the following Theorem.
\begin{mythm}\label{thm:separated}
    Let $g$ be a spatially homogeneous metric. Then there exists a spatial frame of 1-forms $\{\tensor{e}{^\mu}\}$ and coefficients $\tensor{g}{_\mu_\nu}(t)$ such that
    \begin{equation}
        g=\tensor{g}{_\mu_\nu}(t)\,\tensor{e}{^\mu}\tensor{e}{^\nu}.
    \end{equation}
\end{mythm}
\begin{proof}
    Let $\{\tensor{\xi}{_I}\}_{I=1,2,3}$ be those KVFs of $g$ which are purely spatial and form a frame everywhere, which exist due to spatial homogeneity. Since the KLA spanned by $\{\tensor{\xi}{_I}\}$ is a pre-KLA, we can find the IF $\{\tensor{X}{_i}\}_{i=1,2,3}$ by the method of Section \S\ref{sec:table}. Set $\tensor{X}{_0}\equiv\tensor{\partial}{_0}=\tensor{\partial}{_t}$, and then observe that we have
    \begin{equation}
        \tensor{\xi}{_I}[g(\tensor{X}{_\mu},\tensor{X}{_\nu})]=(\mathcal{L}_{\xi_I}g)(\tensor{X}{_\mu},\tensor{X}{_\nu})+g([\tensor{\xi}{_I},\tensor{X}{_\mu}],\tensor{X}{_\nu})+g(\tensor{X}{_\mu},[\tensor{\xi}{_I},\tensor{X}{_\nu}])=0,
    \end{equation}
    ($\mu,\nu=0,1,2,3$), where the first term vanishes due to $\tensor{\xi}{_I}$ being a KVF, and the second and third term due to the IF commuting with KVFs. That $[\tensor{\xi}{_I},\tensor{X}{_0}]=0$ follows from $\tensor{\xi}{_I}$ being purely spatial. Since the KVFs form a frame, this means that the functions $g(\tensor{X}{_\mu},\tensor{X}{_\nu})$ are constant in space, and hence are dependent only on time: $g(\tensor{X}{_\mu},\tensor{X}{_\nu})=g(\tensor{X}{_\mu},\tensor{X}{_\nu})(t)$.

    Given the IF $\{\tensor{X}{_i}\}$, we can uniquely find its dual frame of 1-forms $\{\tensor{e}{^i}\}$. Moreover, with the added vector field $\tensor{X}{_0}=\tensor{\partial}{_t}\leftrightarrow\tensor{e}{^0}=\dd{t}$, we still maintain the desired dual frame relations $\tensor{e}{^\mu}(\tensor{X}{_\nu})=\tensor*{\delta}{^\mu_\nu}$. When we then write the metric in said basis,
    \begin{equation}
        g=g(\tensor{X}{_\mu},\tensor{X}{_\nu})\,\tensor{e}{^\mu}\tensor{e}{^\nu},
    \end{equation}
    we see that we have exactly recovered the form of the metric we wanted to show exists.
\end{proof}

That this can be done is not a new idea \cite[\S9.1]{ryan_homogeneous_1975}, and in some sense defines what it means to be spatially homogeneous (or also known as the synchronous gauge). Yet, we emphasize that the method of Section \S\ref{sec:table}, makes the proof direct and \emph{constructive.} That is, any spatially homogeneous metric, which may be locally presented in \emph{any} (coordinate) frame, can be constructively put into the hypothesized form, i.e.\ with components dependent only on time.


%% file: methodology.tex
\section{Deriving the result}\label{sec:derivation}
In this section we derive the construction that we presented in Section \S\ref{sec:table}. We do so by first laying the groundwork concerning symmetries arising from Lie groups/algebras. Then we introduce spatial homogeneity, and with that context, Bianchi models and the Bianchi classification for 3D Lie algebras. Subsequently, the main idea of the construction, namely \eqref{eq:commutationProperty}, is introduced, after which this equation is solved along the lines of the method of characteristics. We end the section with some auxiliary properties, being uniqueness and diffeomorphism and $\operatorname{GL}(3)$ naturality.

\subsection{Foundations of symmetries}\label{sec:foundations}
\begin{mydef}[Isometry]
    Let $(\mathcal{M}, g)$ a (pseudo-)Riemannian manifold. We say that a diffeomorphism $\phi$ from $\mathcal{M}$ to itself is an \emph{isometry} if the pullback of the metric by $\phi$ leaves it invariant, that is, $\phi^*g=g$. We denote the resulting collection of isometries of a given space with known metric as $\operatorname{Isom}(\mathcal{M})$.
\end{mydef}
The set of all isometries forms a group, acting on the manifold. Since we assumed that $\mathcal{M}$ is simply connected, we forego any global topological considerations. Alternatively, we could say that we assume no isometries to be discrete.

Then, by the Myers-Steenrod theorem, this group is a \emph{Lie} group. As such, we may also consider the infinitesimal version of the isometries, i.e.\ the corresponding elements in the Lie algebra. These are captured by vector fields, which are known as \emph{Killing vector fields} (KVFs). A KVF $\xi$ has the (defining) property that
\begin{equation}\label{eq:define-kvf}
    \mathcal{L}_\xi g=0,
\end{equation}
where $\mathcal{L}$ denotes the Lie derivative. Since all KVFs belonging to a certain metric form a Lie algebra under commutation of vector fields, we may speak of a metric's Killing Lie algebra (KLA).


In general a metric may allow for multiple KVFs; the maximum allowed is $n(n+1)/2$, where $n$ is the dimension of the manifold \cite[\S17.1]{ellis_relativistic_2012}. This also means that the dimension $r$ of the isometry group is bounded by $r\leq n(n+1)/2$. While on a general manifold it is possible that no KVF is allowed for a metric, this is never the case in the class of spacetime manifolds considered in this work: we wish to investigate homogeneous spacetimes, and as we shall see below, this means they have KVFs.

Once we have our infinitesimal descriptions of isometries, we can quantify the orbit of a starting point under these infinitesimal isometries. To this end, we need the following definition.
\begin{mydef}[Transitivity]
    A (sub)group of isometries $G\leq\operatorname{Isom}(\mathcal{M})$ is said to act \emph{transitively} on some submanifold $S$ if $G(s)=S$ \cite[\S5]{ellis_cosmological_2008} for all $s\in S$. In this case $S$ is the orbit of the isometries, also referred to as \emph{surfaces of transitivity.} The dimension of the transitivity subgroup $\dim S$ we denote by $s$.
\end{mydef}
Effectively, through transitivity we infer an equivalence relation on $\mathcal{M}$ of which points can be reached via an isometry. However, there may be ``redundancy'' in the surface of transitivity; we may reach individual points via more than one isometry. We say that $G$ acts \emph{simply transitively} on the surface of transitivity if there is precisely one isometry that connects pairs of points. Or, equivalently, when the infinitesimal generators of isometries are linearly independent as vector fields \cite[\S6.2]{ryan_homogeneous_1975}. Otherwise, we call $G$ \emph{multiply transitive}.

We also need to touch on the concept of \emph{isotropy}. The isotropy subgroup (of isometries) at some $p\in\mathcal{M}$ consists of those isometries which leave the point $p$ unchanged:
\begin{equation}
    \operatorname{Stab}(p):=\{\phi\in\operatorname{Isom}(\mathcal{M}):\phi(p)=p\}.
\end{equation}
The symbol Stab comes from the alternative name for the isotropy subgroup, being the \emph{stabilizer}. Directly, this implies that the infinitesimal generators of those isometries, i.e.\ their associated KVFs, vanish at $p$.
\begin{mylem}
    Let $S$ be a surface of transitivity.
    The dimension of the isotropy subgroup $\dim\operatorname{Stab}(p) \equiv q$ is constant for all points $p\in S$ in the same surface of transitivity.
\end{mylem}
\begin{proof}
    Suppose $a,b \in S$ lie in the same surface of transitivity $S$, and let $\phi$ be an isometry so that $\phi(a)=b$. We notice that we can send an isotropy transformation at $a$ to one at $b$ by noticing
    \begin{equation}
        \begin{split}
            \psi\in\operatorname{Stab}(b)\implies\phi^{-1}\circ \psi\circ\phi\in\operatorname{Stab}(a)&\implies \psi\circ\phi\in\phi\circ\operatorname{Stab}(a) \\
            &\implies\operatorname{Stab}(b)\circ\phi\subseteq\phi\circ\operatorname{Stab}(a)
        \end{split}
    \end{equation}
    By interchanging the roles of $a$ and $b$, the reverse inclusion can also be established, so that we have equality. Hence, $\operatorname{Stab}(a)\cong\operatorname{Stab}(b)$, and so they have the same dimension.
\end{proof}
One also says that the surface of transitivity is \emph{completely anisotropic} if the isotropy subgroup is trivial (and so generated by the zero vector field).

Following the above Lemma, we recognize that $\dim\operatorname{Stab}(S)$ is a well-defined entity for a surface of transitivity $S$. We can then make an important conclusion about the structure of symmetries, by noticing the following: we have that
\begin{equation}\label{eq:total-isom-dim}
    \dim\operatorname{Isom}(S)=\dim\operatorname{Stab}(S)+\dim S
\end{equation}
\cite[\S5.1]{ellis_cosmological_2008}; or, $r=q+s$. Here, $\operatorname{Isom}(S)$ indicates the group of isometries defined on the surface of transitivity. That this is true can be seen heuristically by considering that KVFs (which generate isometries) either vanish at some point on $S$ or they do not. If no KVF vanishes anywhere on $S$, then evidently $\dim\operatorname{Isom}(S)=\dim S$; if not, then $\dim\operatorname{Stab}(S)$ is the maximum number of KVFs that simultaneously vanish at some $p\in S$. Then $S$ must be of dimension $\dim\operatorname{Isom}(S)-\dim\operatorname{Stab}(S)$, as that is how many KVFs at said point still generate the surface of transitivity.

\subsection{Bianchi's classification \& models}\label{sec:classification}
We now narrow our symmetry considerations down to the case of interest, namely that of a 4-dimensional pseudo-Riemannian manifold $(\mathcal{M},g)$.
\begin{mydef}[Spatial homogeneity]
    We say that if $(\mathcal{M},g)$ permits a KLA of dimension 3 transitive on spatial hypersurfaces, then $(\mathcal{M},g)$ is \emph{spatially homogeneous}. We may also shorten this to homogeneous if it is clear from context that spatial homogeneity is meant. A surface of transitivity generated by such a KLA is called a \emph{spatial section}.
\end{mydef}
It is then clear that we should consider the spatial sections as ``space,'' with the remaining dimension fulfilling the role of ``time.'' This gives rise to the split $\mathcal{M}=\mathbb{R}\times S$, a consequence of the assumption of \emph{global hyperbolicity} that is often present in treatises of general relativity (e.g.\ \cite[Thm.\ 5.44]{landsman_foundations_2022}, \cite[\S4.2.1]{gourgoulhon_31_2012}), and implicitly made in the study of cosmology. Using the symmetry language built up in the previous section, we can then articulate precisely what is meant with a Bianchi model.
\begin{mydef}[Bianchi model]\label{def:bianchi-models}
    A \emph{Bianchi model} is a spatially homogeneous spacetime, with a KLA of dimension precisely 3.
\end{mydef}
\begin{myrem}[Differing nomenclature]
    One may also adhere to the definition that a Bianchi model is a spacetime which is spatially homogeneous, without the dimensionality restriction. Under this definition, locally rotationally symmetric (LRS) Bianchi models and Kantowski-Sachs models \cite{kantowski_spatially_1966, goncalves_kantowski-sachs_2022} (which have a KLA of dimension 4) and the FLRW models (dimension 6) will be special cases of Bianchi models, instead of being completely separated. Since for our purposes we wish to restrict to three-dimensional KLAs, we make the additional dimensionality requirement.
\end{myrem}
In view of equation \eqref{eq:total-isom-dim}, it follows immediately that a Bianchi model is completely anisotropic. Furthermore, the vector fields which compose the KLA \emph{define a frame} for the tangent space $T_pS$ for all $p\in S$: indeed, if somewhere they did not define a frame, then locally the spatial section would be of a dimension smaller than 3, which is not possible due to \eqref{eq:total-isom-dim}.
\begin{myrem}[Different choices $q$, $s$ in \eqref{eq:total-isom-dim}]
    Effectively in the definition of Bianchi models we prescribed dimensions for the isotropy subgroup and surfaces of transitivity. Since mathematically there is nothing special about these choices, others can be considered as well. For instance, \cite[\S5.2]{ellis_cosmological_2008} presents a table outlining the valid choices of dimensions, and the names of the resulting models (which includes the aforementioned LRS Bianchi models and Kantowski-Sachs models).
\end{myrem}

Now that we have defined what a Bianchi model is, we may introduce a further level of detail by invoking the Lie algebraic structure of the KLA. Since for Bianchi models the dimension of the KLA is 3, a Bianchi model is effectively characterized by a 3-dimensional Lie algebra --- we shall write its basis as $\{\tensor{\xi}{_I}\}_{I=1,2,3}$. Note that here the subscripts do not indicate coordinate dependence, but are simply the labels of the vector fields.

There are infinitely many Lie algebras of dimension 3, but if we construct equivalence classes through $\operatorname{GL}(3,\mathbb{R})$-invariance of the Lie algebra's basis, it turns out that there are only nine distinct classes of Lie algebras \cite{bianchi1898}. One could classify these in various ways that one sees fit, and depending on the intended use. In the astrophysics literature it is common to consider the \emph{Bianchi(-Behr)} classification of Lie algebras. This is given in Table \ref{tab:bianchis}, with additional mathematical information \cite{ellis_bianchi_2006, patera_invariants_1976, jacobson_lie_1961}.
\begin{table}[h!]
    \centering
    \renewcommand{\arraystretch}{1.5}
    \begin{tabular}{c|c|c|c|c||c|c}
         $a$ & $n_1$ & $n_2$ & $n_3$ & Bianchi type & $\operatorname{dim}[\mathfrak{g},\mathfrak{g}]$ & comments \\
        \hline\hline
        0 & 0 & 0 & 0 & I & 0 & abelian; $\mathfrak{i}^{\oplus3}$ \\
        \hline
        0 & 1 & 0 & 0 & II & 1 & nilpotent \\
        \hline
        0 & 0 & 1 & -1 & VI$_0$ & 2 & \\
        \hline
        0 & 0 & 1 & 1 & VII$_0$ & 2 & \\
        \hline
        0 & -1 & 1 & 1 & VIII & 3 & $\operatorname{sl}(2)$\\
        \hline
        0 & 1 & 1 & 1 & IX & 3 & $\operatorname{su}(2)\cong\operatorname{so}(3)$ \\
        \hline\hline
        1 & 0 & 1 & -1 & III & 1 & $\mathfrak{l}(2)\oplus\mathfrak{i}$ \\
        \hline
        1 & 0 & 0 & 1 & IV & 2 & \\
        \hline
        1 & 0 & 0 & 0 & V & 2 & \\
        \hline
        $\sqrt{-h}$ & 0 & 1 & -1 & VI$_h$ & 2 & $h<0$ \\
        \hline
        $\sqrt{h}$ & 0 & 1 & 1 & VII$_h$ & 2 & $h>0$ \\
    \end{tabular}
    \caption{Table of Bianchi types, together with further mathematical detail. Here, $\mathfrak{i}$ refers to the 1-dimensional Lie algebra, and $\mathfrak{l}(2)$ to the irreducible 2-dimensional Lie algebra.}
    \label{tab:bianchis}
\end{table}
The structure constants are determined from the table's entries as
\begin{equation}\label{eq:struc-constants}
    [\tensor{\xi}{_1},\tensor{\xi}{_2}]=n_3\tensor{\xi}{_3}-a\tensor{\xi}{_2},\quad[\tensor{\xi}{_2},\tensor{\xi}{_3}]=n_1\tensor{\xi}{_1},\qand[\tensor{\xi}{_3},\tensor{\xi}{_1}]=n_2\tensor{\xi}{_2}+a\tensor{\xi}{_3}.
\end{equation}
That this indeed yields a valid classification (i.e.\ every Lie algebra belongs to precisely one of these types by a suitable $\operatorname{GL}(3)$ transformation of its basis) is shown in \cite{ellis_class_1969}. The classes are thus composed of Lie algebras which, when their basis is suitably chosen, have the prototypical structure constants.

We are then able to articulate what we mean when we consider metrics of Bianchi type $n$.
\begin{mydef}[Bianchi type $n$ metric]
    A spacetime model $(\mathcal{M},g)$ (and in particular its metric) is of \emph{Bianchi type} $n$ if it is
    \begin{enumerate}
        \item spatially homogeneous,
        \item its KLA has dimension precisely 3, and
        \item the KLA is of Bianchi type $n$, i.e.\ it permits a basis with structure constants as found through Table \ref{tab:bianchis} and equation \eqref{eq:struc-constants}, for Bianchi type $n$.
    \end{enumerate}
    We call the metric $g$ of such a Bianchi model a Bianchi type $n$ metric.
\end{mydef}
For further details, see \cite[\S2.7.3\,\&\,\S17.1.3]{ellis_relativistic_2012} and \cite[\S5]{ellis_cosmological_2008}.

\subsection{To frame a Killing}\label{sec:findingFrame}
The above theory will be applied hinging on the following observation. Let $g$ be a spatially homogeneous metric, with nowhere-vanishing KVFs $\{\tensor{\xi}{_I}\}_{I=1}^3$ which form a frame for space everywhere. Now suppose we have another frame for space $\{\tensor{X}{_i}\}_{i=1}^3$ --- and so its coframe $\{\tensor{e}{^i}\}$ --- satisfying the property
\begin{equation}\label{eq:commutationProperty}
    [\tensor{\xi}{_I},\tensor{X}{_i}]=0.
\end{equation}
Owing to this property, let us call $\{\tensor{X}{_i}\}$ the \emph{invariant frame} (IF), and $\{\tensor{e}{^i}\}$ the \emph{dual frame} (DF). Then, we can draw the following conclusions.
\begin{enumerate}
    \item \emph{The frame $\{\tensor{X}{_i}\}$ is also a Lie algebra.} Let $\tensor{D}{^k_i_j}(\vb*{x})$ be the \emph{functions} such that $[\tensor{X}{_i},\tensor{X}{_j}]=\tensor{D}{^a_i_j}\tensor{X}{_a}$. Then Lie differentiating:
    \begin{equation}
        \mathcal{L}_{\tensor{\xi}{_I}}\left(\tensor{D}{^a_i_j}\tensor{X}{_a}\right)=(\tensor{\xi}{_I}\tensor{D}{^a_i_j})\tensor{X}{_a}\qand\mathcal{L}_{\tensor{\xi}{_I}}[\tensor{X}{_i},\tensor{X}{_j}]=[\tensor{\xi}{_I},[\tensor{X}{_i},\tensor{X}{_j}]]\overset{\star}{=}0,
    \end{equation}
    where $\star$ is due to the Jacobi identity. Hence, $\tensor{\xi}{_I}\tensor{D}{^k_i_j}=0$, in particular for all choices of $I$ given fixed $i,j,k$. Since the $\{\tensor{\xi}{_I}\}$ are spanning, the $\tensor{D}{^k_i_j}(\vb*{x})$ must in fact be constant. Hence, the $\{\tensor{X}{_i}\}$ have constant commutation relations, and thus form a Lie algebra. The $\tensor{D}{^k_i_j}$ are then the structure constants.
    \item \emph{The Lie derivative of the coframe vanishes.} Indeed, since
    \begin{equation}
        0=\mathcal{L}_{\tensor{\xi}{_I}}\langle\tensor{e}{^i},\tensor{X}{_j}\rangle=\langle\mathcal{L}_{\tensor{\xi}{_I}}\tensor{e}{^i},\tensor{X}{_j}\rangle+\cancel{\langle\tensor{e}{^i},\mathcal{L}_{\tensor{\xi}{_I}}\tensor{X}{_j}\rangle}
    \end{equation}
    and the $\{\tensor{X}{_i}\}$ are spanning, $\mathcal{L}_{\tensor{\xi}{_I}}\tensor{e}{^i}$ vanishes for all choices of $i,k$.
\end{enumerate}
In particular from the second observation, we see that objects of the form $\tensor{\gamma}{_i_j}\tensor{e}{^i}\otimes\tensor{e}{^j}$ for \emph{constant} $\tensor{\gamma}{_i_j}$ vanish under Lie differentiation by the $\{\tensor{\xi}{_I}\}$. By choosing the coefficients $\{\tensor{\gamma}{_i_j}\}$ symmetrically and in a positive definite fashion, $\tensor{\gamma}{_i_j}\tensor{e}{^i}\tensor{e}{^j}$ becomes a Riemannian metric on $\Sigma^3$ --- now with known KVFs.

Hinging on \eqref{eq:commutationProperty} being satisfied, we can thus find spatial metrics on which a pre-KLA is indeed realized as a KLA. In order to promote this construction to the spacetime case, we leverage global hyperbolicity by simply ``adding a time axis'' as an additional parameter . Let $t:\mathcal{M}\to\mathbb{R}$ be a function whose existence is afforded to us by global hyperbolicity. Then set $\tensor{X}{_0}=\pdv{t}\iff\tensor{e}{^0}=\dd{t}$, and consider now the spacetime (co-)frames $\{\tensor{X}{_\mu}\}_{\mu=0}^3$ and $\{\tensor{e}{^\mu}\}_{\mu=0}^3$. The additional commutation relation is $[\tensor{X}{_0},\tensor{X}{_i}]=0$, as the $\{\tensor{X}{_i}\}$ should be independent of $t$. The reasoning as above can be repeated, but now for tensors of the form $\tensor{\gamma}{_\mu_\nu}(t)\,\tensor{e}{^\mu}\tensor{e}{^\nu}$: the added dependence of the coefficients on $t$ has no effect, as the purely spatial pre-KLA $\{\tensor{\xi}{_I}\}$ does not ``see'' the time dependence.

The freedom one has on how to choose the cosmic time function $t$ should also be reflected in this formalism. It is found by doing a time-dependent rotation of $\{\tensor{X}{_\mu}\}$, and thereby specifically of $\tensor{X}{_0}$, which is the gradient of $t$: by altering this, effectively the ``time axis'' is altered. It may be done by defining the new frame $\{\tensor{Y}{_\mu}\}$ such that $\tensor{Y}{_\mu}\!(t)=\tensor{y}{^\alpha_\mu}(t)\,\tensor{X}{_\alpha}$, and then ensuring that the relations $[\tensor{Y}{_0},\tensor{Y}{_i}]=0$ and $[\tensor{Y}{_i},\tensor{Y}{_j}]=\tensor{D}{^a_i_j}\tensor{Y}{_a}$ still hold. If the $\tensor{y}{^\mu_0}(t)$ are specified (i.e.\ the new time axis chosen), the $\tensor{y}{^\mu_i}(t)$ can be chosen so that both relations are satisfied --- see \cite[\S6.3]{ryan_homogeneous_1975} for details.

Summarizing: when given a pre-KLA $\{\tensor{\xi}{_I}\}$, if we can find a frame $\{\tensor{X}{_i}\}$ so that \eqref{eq:commutationProperty} is satisfied, then a metric written in its coframe with coefficients depending only on time will have the pre-KLA as its KLA. Or, more symbolically: given the pre-KLA $\{\tensor{\xi}{_I}\}$,
\begin{equationFrame}[frametitle = \colorbox{white}{Invariant frame is desired metric frame}]
    \begin{equation}\label{eq:basisInsight}
        \{\tensor{X}{_i}\}:\text{\eqref{eq:commutationProperty} holds}\implies \{\tensor{\xi}{_I}\}\text{ are Killing on }g=\tensor{g}{_\mu_\nu}(t)\,\tensor{e}{^\mu}\tensor{e}{^\nu},
    \end{equation}
\end{equationFrame}
where $\langle\tensor{e}{^\mu},\tensor{X}{_\nu}\rangle=\tensor*{\delta}{^\mu_\nu}$ and $\tensor{e}{^0}=\dd{t}$. Hence, it is clear that we need to investigate when/how this first equality \eqref{eq:commutationProperty} is satisfied or not. Recalling that $\{\tensor{\xi}{_I}\}$ is a frame, let us expand the $\{\tensor{X}{_i}\}$ in that basis and write down the relation \eqref{eq:commutationProperty}:
\begin{equation}
    \begin{split}
        \tensor{X}{_i}=\tensor*{X}{_i^A}(\vb*{x})\,\tensor{\xi}{_A}\implies[\tensor{\xi}{_I},\tensor*{X}{_i^A}\tensor{\xi}{_A}]=0&\implies\left(\tensor{\xi}{_I}\tensor{X}{_i^A}+\tensor{X}{_i^B}\tensor{C}{^A_I_B}\right)\tensor{\xi}{_A}=0 \\
        &\implies\tensor{\xi}{_I}\tensor{X}{_i^J}=-\tensor{C}{^J_I_A}\tensor{X}{_i^A}.
    \end{split}
\end{equation}
The above should hold for all choices of $i,I,J$. We recognize this to be a linear differential equation in the \emph{components} $\tensor{X}{_i^I}$, with \emph{directional derivatives} $\tensor{\xi}{_I}$.

The solution strategy follows along the lines of the method of characteristics \cite[\S3.2]{evans_partial_2010}. Assume that we parametrize by $s$ a path in local coordinates $(\tensor{x}{^1},\tensor{x}{^2},\tensor{x}{^3})$, whose tangent is given by the vector $T(\vb*{x})=\tensor{\tilde{T}}{^a}(\vb*{x})\tensor{\partial}{_a}=\tensor{T}{^A}(\vb*{x})\tensor{\xi}{_A}$, i.e.\ so that $\dv{\tensor{x}{^i}}{s}=\tensor{\tilde{T}}{^i}$. Then, examining the change of $\tensor{X}{_i^j}$ \emph{along this path,} we find the following equation:
\begin{equation}\label{eq:equationAlongFlows}
    \dv{\tensor{X}{_i^I}}{s}=\pdv{\tensor{X}{_i^I}}{\tensor{x}{^a}}\dv{\tensor{x}{^a}}{s}=\tensor{\tilde{T}}{^a}\tensor{\partial}{_a}\tensor{X}{_i^I}=\tensor{T}{^A}\tensor{\xi}{_A}\tensor{X}{_i^I}=-\tensor{C}{^I_A_B}\tensor{T}{^A}\tensor{X}{_i^B}.
\end{equation}
The evolution of the components $\tensor{X}{_i^j}$ along this path is therefore a 9 (3$\times$3) dimensional first order ODE. Therefore, by standard arguments from ODE theory, a solution exists and is unique provided that we supply an initial condition.

The initial condition we choose is to demand that at some $p\in\Sigma^3$, we have that $\tensor{X}{_i}|_p=\tensor{\xi}{_i}|_p$, i.e.\ that $\tensor{X}{_i^I}|_p=\tensor*{\delta}{^I_i}$. Namely, following \cite[\S6.3]{ryan_homogeneous_1975}, we will then have that
\begin{equation}
    \tensor{D}{^k_i_j}\tensor{X}{_k^A}\tensor{\xi}{_A}=\tensor{D}{^k_i_j}\tensor{X}{_k}=[\tensor{X}{_i},\tensor{X}{_j}]=[\tensor{X}{_i^A}\tensor{\xi}{_A},\tensor{X}{_j^B}\tensor{\xi}{_B}]=\tensor{X}{_j^A}\tensor{X}{_i^C}\tensor{C}{^B_A_C}\tensor{\xi}{_B},
\end{equation}
so that
\begin{equation}
    \tensor{D}{^k_i_j}\tensor{X}{_k^I}=\tensor{X}{_j^A}\tensor{X}{_i^B}\tensor{C}{^I_A_B}.
\end{equation}
Since both $\tensor{D}{^k_i_j}$ and $\tensor{C}{^K_I_J}$ are sets of structure constants, we can find one in terms of the other by simply looking at a single point; doing so at $p$ yields $\tensor{D}{^k_i_j}=\tensor{C}{^k_j_i}=-\tensor{C}{^k_i_j}$. This also shows that the IF $\{\tensor{X}{_i}\}$ and pre-KLA $\{\tensor{\xi}{_I}\}$ are isomorphic as Lie algebras.

Let us make two additional remarks.
\begin{enumerate}
    \item We can choose to examine the especially simple case $\tensor{T}{^I}=\tensor*{\delta}{^I_J}$, for some given $J$. Essentially, this follows the flow only of the vector field $\tensor{\xi}{_J}$. This flow is called the \emph{characteristic} of $\tensor{\xi}{_J}$.
    \item Assuming the above remark, we can then view the equation \eqref{eq:equationAlongFlows} also as a matrix differential equation with constant entries. Define $\tensor{\vec{X}}{_i}:=\pmqty{\tensor{X}{_i^1} & \tensor{X}{_i^2} & \tensor{X}{_i^3}}^\text{T}$, so that \eqref{eq:equationAlongFlows} becomes
    \begin{equation}\label{eq:evolutionAlongCharacteristic}
        \dv{s}\tensor{\vec{X}}{_i}=-\vb{C}_I\tensor{\vec{X}}{_i}\implies\tensor{\vec{X}}{_i}(s)=\exp(-\vb{C}_Is)\tensor{\vec{X}}{_i_{,0}},
    \end{equation}
    where $\vb{C}_I$ is a matrix with entries $(\vb{C}_I)^J_K=\tensor{C}{^J_I_K}$. So, whenever we evolve the components $\tensor{X}{_i^I}$ along the characteristic $\tensor{\xi}{_I}$, we can utilize formula \eqref{eq:evolutionAlongCharacteristic}. The $\vb{C}_I$ is the image of $\tensor{\xi}{_I}$ under the adjoint representation of the pre-KLA -- see also Remark~\ref{rem:jantzen}.
\end{enumerate}

With these two remarks, we can finally create a recipe.
\begin{textFrame}[frametitle = \colorbox{white}{Invariant frame from pre-KLA}]
    Given a pre-KLA $\{\tensor{\xi}{_I}\}$, let $\{\tensor{X}{_i}\}$ be its IF, with $p$ such that $\tensor{X}{_i^I}|_p=\tensor{\delta}{_i^I}$. In order to find the components of the IF $\tensor{X}{_i^I}$ at $q\neq p$, follow these steps.
    \begin{enumerate}
        \item Find a path connecting $p$ to $q$ consisting of at most three segments, where each segment follows a characteristic for some length.
        \item With $s_I$ as the length of each characteristic, find the matrices $\exp(-\vb{C}_Is_I)$ (no sum).
        \item Left-multiply the matrix of the first segment with that of the second, then left-multiply with that of the third.
        \item The first column of the matrix so obtained yields $\tensor{\vec{X}}{_1}|_q$, the second $\tensor{\vec{X}}{_2}|_q$, and the third $\tensor{\vec{X}}{_3}|_q$.
    \end{enumerate}

\end{textFrame}
Additionally, if we view the various $s_I$ as functions $s_I(q)$, then we have a method by which we can find the value of the invariant frame at any $q$. However, the rub is then that the functions $s_I(q)$ depend non-trivially on $q$. Given an order of characteristics, this non-triviality corresponds to finding the lengths for which you need to follow them in order to arrive at the point of interest.

\begin{figure}[b]
    \centering
    \tikzset{every picture/.style={line width=0.5pt}} 
    \begin{tikzpicture}[x=0.75pt,y=0.75pt,yscale=-.75,xscale=.75]

        \draw  (175.71,355.52) -- (344.36,355.52)(192.58,217.34) -- (192.58,370.87) (337.36,350.52) -- (344.36,355.52) -- (337.36,360.52) (187.58,224.34) -- (192.58,217.34) -- (197.58,224.34)  ;
        \draw    (192.58,355.52) -- (148.72,396.99) -- (102.76,440.45) ;
        \draw [shift={(101.31,441.82)}, rotate = 316.6] [color={rgb, 255:red, 0; green, 0; blue, 0 }  ][line width=0.75]    (10.93,-3.29) .. controls (6.95,-1.4) and (3.31,-0.3) .. (0,0) .. controls (3.31,0.3) and (6.95,1.4) .. (10.93,3.29)   ;
        \draw [color={rgb, 255:red, 0; green, 0; blue, 0 }  ,draw opacity=1 ][line width=2.25]  [dash pattern={on 2.53pt off 3.02pt}]  (192.58,355.52) .. controls (279.11,355.9) and (379.76,437.54) .. (444.78,400.76) ;
        \draw [shift={(448.71,398.36)}, rotate = 146.73] [fill={rgb, 255:red, 0; green, 0; blue, 0 }  ,fill opacity=1 ][line width=0.08]  [draw opacity=0] (14.29,-6.86) -- (0,0) -- (14.29,6.86) -- cycle    ;
        \draw [color={rgb, 255:red, 0; green, 0; blue, 0 }  ,draw opacity=1 ][line width=2.25]  [dash pattern={on 6.75pt off 4.5pt}]  (448.71,398.36) .. controls (431.15,359.49) and (371.7,251.22) .. (288.93,214.22) ;
        \draw [shift={(285.13,212.58)}, rotate = 22.63] [fill={rgb, 255:red, 0; green, 0; blue, 0 }  ,fill opacity=1 ][line width=0.08]  [draw opacity=0] (14.29,-6.86) -- (0,0) -- (14.29,6.86) -- cycle    ;
        \draw [color={rgb, 255:red, 0; green, 0; blue, 0 }  ,draw opacity=1 ][line width=2.25]    (285.13,212.58) .. controls (352.22,185.76) and (407.68,181.41) .. (470.93,191.58) .. controls (503.61,196.84) and (538.36,205.98) .. (577.88,217.9) ;
        \draw  [color={rgb, 255:red, 0; green, 0; blue, 0 }  ,draw opacity=1 ][fill={rgb, 255:red, 0; green, 0; blue, 0 }  ,fill opacity=1 ] (579.38,210.4) -- (582.02,215.76) -- (587.94,216.62) -- (583.66,220.79) -- (584.67,226.68) -- (579.38,223.9) -- (574.09,226.68) -- (575.1,220.79) -- (570.82,216.62) -- (576.73,215.76) -- cycle ;
        \draw  [fill={rgb, 255:red, 0; green, 0; blue, 0 }  ,fill opacity=1 ] (186.58,355.52) .. controls (186.58,352.2) and (189.27,349.52) .. (192.58,349.52) .. controls (195.89,349.52) and (198.58,352.2) .. (198.58,355.52) .. controls (198.58,358.83) and (195.89,361.52) .. (192.58,361.52) .. controls (189.27,361.52) and (186.58,358.83) .. (186.58,355.52) -- cycle ;
        \draw  [color={rgb, 255:red, 0; green, 0; blue, 0 }  ,draw opacity=1 ] (206.88,369.9) .. controls (205.41,374.33) and (206.89,377.28) .. (211.32,378.75) -- (299.18,407.89) .. controls (305.51,409.99) and (307.94,413.26) .. (306.47,417.69) .. controls (307.94,413.26) and (311.84,412.09) .. (318.17,414.19)(315.32,413.25) -- (406.03,443.33) .. controls (410.46,444.8) and (413.41,443.32) .. (414.88,438.89) ;
        \draw  [color={rgb, 255:red, 0; green, 0; blue, 0 }  ,draw opacity=1 ] (465.88,376.9) .. controls (469.73,374.27) and (470.33,371.02) .. (467.7,367.17) -- (419.81,297.2) .. controls (416.04,291.7) and (416.09,287.63) .. (419.94,284.99) .. controls (416.09,287.63) and (412.28,286.2) .. (408.51,280.69)(410.2,283.17) -- (360.62,210.72) .. controls (357.98,206.87) and (354.74,206.26) .. (350.89,208.9) ;
        \draw  [color={rgb, 255:red, 0; green, 0; blue, 0 }  ,draw opacity=1 ] (568.88,177.9) .. controls (568.88,173.23) and (566.55,170.9) .. (561.88,170.9) -- (441.88,170.9) .. controls (435.21,170.9) and (431.88,168.57) .. (431.88,163.9) .. controls (431.88,168.57) and (428.55,170.9) .. (421.88,170.9)(424.88,170.9) -- (301.88,170.9) .. controls (297.21,170.9) and (294.88,173.23) .. (294.88,177.9) ;

        \draw (370.08,365.57) node [anchor=north west][inner sep=0.75pt]  [font=\large,color={rgb, 255:red, 0; green, 0; blue, 0 }  ,opacity=1 ]  {$\xi _{1}$};
        \draw (311.28,246.04) node [anchor=north west][inner sep=0.75pt]  [font=\large,color={rgb, 255:red, 0; green, 0; blue, 0 }  ,opacity=1 ]  {$\xi _{2}$};
        \draw (504.29,209.31) node [anchor=north west][inner sep=0.75pt]  [font=\large,color={rgb, 255:red, 0; green, 0; blue, 0 }  ,opacity=1 ]  {$\xi _{3}$};
        \draw    (133.17,151.95) -- (133.17,187.95) -- (94.17,187.95)  ;
        \draw (97.17,156.35) node [anchor=north west][inner sep=0.75pt]  [font=\Large]  {$\mathbb{R}^{3}$};
        \draw (154.42,326.35) node [anchor=north west][inner sep=0.75pt]  [font=\large]  {$p$};
        \draw (594.98,207.91) node [anchor=north west][inner sep=0.75pt]  [font=\large]  {$q$};
        \draw (290,416.4) node [anchor=north west][inner sep=0.75pt]  [font=\large,color={rgb, 255:red, 0; green, 0; blue, 0 }  ,opacity=1 ]  {$s_{1}$};
        \draw (422,266.4) node [anchor=north west][inner sep=0.75pt]  [font=\large,color={rgb, 255:red, 0; green, 0; blue, 0 }  ,opacity=1 ]  {$s_{2}$};
        \draw (422,133.4) node [anchor=north west][inner sep=0.75pt]  [font=\large,color={rgb, 255:red, 0; green, 0; blue, 0 }  ,opacity=1 ]  {$s_{3}$};

    \end{tikzpicture}
    \caption{Image displaying the construction. We have the path from $p$ to $q$, following the various characteristics, indicated by different dashing patterns, for the specified lengths $s_I$, $I=1,2,3$. In this case, we follow the characteristic of $\xi_1$ for distance $s_1$, then the characteristic of $\xi_2$ for $s_2$, and finally the characteristic of $\xi_3$ for $s_3$, ending up at $q$. Along each characteristic we know how to evaluate the invariant frame components $\tensor{X}{_i^I}$, namely by means of \eqref{eq:evolutionAlongCharacteristic}.}
    \label{fig:constructionDiagram}
\end{figure}

\begin{myexample}[Finding the Bianchi II invariant frame]\label{ex:bianchi2Computation}
    As we learn from Table \ref{tab:bianchis}, the canonical structure constants for a Bianchi II Lie algebra are given by $\tensor{C}{^1_2_3}=-\tensor{C}{^1_3_2}=1$. Hence, the matrices we need for step (ii) are
    \begin{equation}
        \vb{C}_1=0,\quad\vb{C}_2=\pmqty{0 & 0 & 1 \\ 0 & 0 & 0 \\ 0 & 0 & 0}\qand\vb{C}_3=\pmqty{0 & -1 & 0 \\ 0 & 0 & 0 \\ 0 & 0 & 0}.
    \end{equation}
    Suppose the coordinate origin is our base point $p$. Then the durations $s_I$ can be considered as functions on $s_I=s_I(\vb*{x})$. Consequently,
    \begin{equation}
      \begin{split}
        &\exp(-\vb{C}_1s_1)=I,\quad\exp(-\vb{C}_2s_2)=\pmqty{1 & 0 & -s_2 \\ 0 & 1 & 0 \\ 0 & 0 & 1},\\
        &\mbox{and}\quad\exp(-\vb{C}_3s_3)=\pmqty{1 & s_3 & 0 \\ 0 & 1 & 0 \\ 0 & 0 & 1}.
      \end{split}
    \end{equation}
    We follow our path first along the characteristic of $\tensor{\xi}{_1}$, then $\tensor{\xi}{_2}$, and finally $\tensor{\xi}{_3}$ (the $\{\tensor{\xi}{_I}\}$ are the basis vector fields for the Lie algebra corresponding to the structure constants). So, the total matrix computed is
    \begin{equation}
        \exp(-\vb{C}_3s_3)\exp(-\vb{C}_2s_2)\exp(-\vb{C}_1s_1)=\pmqty{1 & s_3 & -s_2 \\ 0 & 1 & 0 \\ 0 & 0 & 1}.
    \end{equation}
    The components $\tensor{X}{_i^I}$ can now be read as the columns, so that the IF vector fields become
    \begin{equation}\label{eq:bianchi2ExampleIF}
        \tensor{X}{_1}=1,\quad\tensor{X}{_2}=s_3\,\tensor{\xi}{_1}+\tensor{\xi}{_2}\qand\tensor{X}{_3}=-s_2\,\tensor{\xi}{_1}+\tensor{\xi}{_3}.
    \end{equation}

    Now we just need to find the functions $s_2$ and $s_3$ in terms of the given pre-KLA. The simplest way to do this is to use \eqref{eq:bianchi2ExampleIF} as an ansatz, and then choose $s_2$ and $s_3$ so that \eqref{eq:commutationProperty} is satisfied. Carrying this out, we derive the properties
    \begin{equation}
        s_2:\begin{cases}\tensor{\xi}{_1}s_2=0 \\ \tensor{\xi}{_2}s_2=1 \\ \tensor{\xi}{_3}s_2=0\end{cases}\qand s_3:\begin{cases}\tensor{\xi}{_1}s_3=0 \\ \tensor{\xi}{_2}s_3=0 \\ \tensor{\xi}{_3}s_3=1 \end{cases}.
    \end{equation}
    At this stage, one would need to choose an explicit coordinate representation of $\{\tensor{\xi}{_I}\}$ in order to concretely determine $s_2$ and $s_3$ as functions of those coordinates.
\end{myexample}

\begin{myrem}[Interpretation with respect to abstract Lie group theory]\label{rem:jantzen}
    
    The matrix-exponential relation in \eqref{eq:evolutionAlongCharacteristic} has a natural interpretation in terms of Lie group theory, and indeed, this is the point of view that drives Jantzen's argument~\cite{jantzen_dynamical_1979}. The collection $\{\tensor{\xi}{_I}\}$, disregarding its nature as vector fields, can be viewed as a basis for the Lie algebra associated to the Lie group of isometries for $\Sigma^3$, $\mathfrak{g}:=\operatorname{Lie}(\operatorname{Isom}\Sigma^3)$. Then the matrices $\vb{C}_I$ that appear in \eqref{eq:equationAlongFlows}-\eqref{eq:evolutionAlongCharacteristic} are the matrices of the adjoint representation of $\mathfrak{g}$, with basis $\{\tensor{\xi}{_I}\}$: $(\vb{C}_I)^J_K=\tensor{C}{^J_I_K}$. In this sense, the action described by \eqref{eq:evolutionAlongCharacteristic} denotes the action of the adjoint representation, acting upon the Lie algebra element $X_i$.
    Therefore, constructing an invariant frame by solving the system of PDEs \eqref{eq:commutationProperty} corresponds to choosing a left-invariant frame on the Lie group $G$ associated with $\mathfrak{g}$.

    The corresponding dual coframe $\{\tensor{e}{^i}\}$ obtained in Section~\ref{sec:classification} then can be seen as the Maurer-Cartan coframe, satisfying the Maurer-Cartan equation
    \begin{equation}
        \dd{\tensor{e}{^i}}=-\tfrac{1}{2}\tensor{C}{^i_a_b}\tensor{e}{^a}\wedge\tensor{e}{^b}.
    \end{equation}
    In other words, the algorithm presented here can be viewed as an explicit coordinate-level method for recovering the Maurer-Cartan structure in a general chart, without requiring the reader to invoke abstract group-theoretic machinery. This also provides another explanation of the path-independence of the construction, to be explored in the next subsection: the group structure guarantees that composing flows of the $\tensor{\xi}{_I}$ reproduces the group multiplication law, while the adjoint representation ensures consistency across different orderings of flows. In this sense, our algorithm provides a concrete way to compute the Maurer–Cartan coframe in any coordinate chart.

\end{myrem}

\subsection{Auxiliary results \& classification}\label{sec:auxillaries}
Now that we have a method by which we can find the invariant frame, we can also investigate the procedure further in order to derive some additional results.
\begin{description}
    \item[Uniqueness] Due to the nature of the IF being determined in Section \S\ref{sec:findingFrame} by means of an IVP, we already have uniqueness of the solution along a given (characteristic) curve. However, this does not guarantee that evolution along two different curves to the same end point will result in the same IF at that end point: there could be path-dependency, leading to two or more IFs belonging to the same pre-KLA.

    We show that there is no path-dependency, so that \emph{the IF determined in \S\ref{sec:findingFrame} is unique.} To do so, let us by contradiction assume that belonging to a pre-KLA $\{\tensor{\xi}{_I}\}$ there are two IFs, $\{\tensor{W}{_i}\}$ and $\{\tensor{Y}{_i}\}$, which both satisfy $W_i|_p=\tensor{\xi}{_i}|_p=\tensor{Y}{_i}|_p$ at some $p\in\Sigma^3$. Then, $\tensor{Z}{_i}=\tensor{W}{_i}-\tensor{Y}{_i}$, $i=1,2,3$, are vector fields (i) which also commute with all $\tensor{\xi}{_I}$, so equation \eqref{eq:evolutionAlongCharacteristic} holds, and (ii) so that $\tensor{Z}{_i}|_p=0$. Now let us write down \eqref{eq:evolutionAlongCharacteristic}:
    \begin{equation}
        \dv{\tensor{Z}{_i^I}}{s}=-\tensor{C}{^I_A_B}\tensor{T}{^A}\tensor{Z}{_i^B}.
    \end{equation}
    Starting at $p$, following \emph{any} path the components $\tensor{Z}{_i^I}$ will not change, as the derivative depends on the value at $p$, which is zero. Therefore, we can conclude that $\tensor{Z}{_i}=0$ as vector field, and thus that $\tensor{W}{_i}=\tensor{Y}{_i}$ for all $i$. Hence, the IF with given base point $p$, is unique.

    It is important to stress that this result crucially relies on the simply-connectedness of $\Sigma^3$. When this assumption fails, topological obstruction can appear, caused, for instance, by the presence of discrete symmetries.



    \medskip
    \item[Diffeomorphism naturality] Suppose now that we know the pre-KLA and IF for some situation, and that we have a transformation from the one pre-KLA to a new one. How does the IF for the new pre-KLA relate to the old one? In order to assess this, let $\{\tensor{\xi}{_I}\}$ and $\{\tensor{X}{_i}\}$ be a pre-KLA and IF on $\mathcal{M}$, respectively, $\{\tensor{{\xi'}}{_I}\}$ a pre-KLA on $\mathcal{N}$, and $\phi$ a mapping so that $\phi:\tensor{\xi}{_I}\mapsto\tensor{{\xi'}}{_I}$.

    If we have $\phi=\Phi_*$ for some diffeomorphism $\Phi:\mathcal{M}\to\mathcal{N}$, then the IF belonging to $\{\tensor{{\xi'}}{_I}\}$ is given by $\{\phi(\tensor{X}{_i})\}$. To see this, observe firstly that
    \begin{equation}
        [\tensor{{\xi'}}{_I},\phi(\tensor{X}{_i})]=[\Phi_*\tensor{\xi}{_I},\Phi_*\tensor{X}{_i}]=\Phi_*[\tensor{\xi}{_I},\tensor{X}{_i}]=0
    \end{equation}
    so commutativity is established, and secondly that
    \begin{equation}
        \Phi_*\tensor{X}{_i}|_\Phi(p)f=\tensor{X}{_i}|_p(f\circ\Phi)=\tensor{\delta}{_i^I}\tensor{\xi}{_I}|_p(f\circ\Phi)=\tensor{\delta}{_i^I}\Phi_*\tensor{\xi}{_I}|_pf,
    \end{equation}
    so that at the image of $p$, the $\{\phi(\tensor{X}{_i})\}$ are exactly of the desired form. Finally, since we showed above that the IF is unique, we must conclude that indeed $\{\Phi_*\tensor{X}{_i}\}$ is the IF corresponding to the pre-KLA $\{\Phi_*\tensor{\xi}{_I}\}$.

    \medskip
    \item[$\operatorname{GL}(3)$ naturality] Suppose now that we perform a \emph{constant} $\operatorname{GL}(3)$ transformation $R$ on a pre-KLA $\{\tensor{\xi}{_I}\}$, so $\tensor{\xi}{_I}\mapsto\tensor{{\xi'}}{_I}=\tensor{R}{^A_I}\tensor{\xi}{_A}$. It is straightforward to check that $\{\tensor{R}{^a_i}\tensor{X}{_a}\}$ is the new IF. Additionally, the components $\tensor{X}{_i^I}$ will transform as $\tensor{{X'}}{_i^I}=\tensor{R}{^a_i}\tensor{X}{_a^A}\tensor{(R^{-1})}{^I_A}$.
\end{description}

Taking all of the above together, we conclude it is sensible to calculate an IF for a class representative of pre-KLA viewed as Lie algebra, i.e.\ with given structure constants, and then tabulate the results over the different Lie algebras sorted by Bianchi type. This is what has been done for Table \ref{tab:forms}, and that we have already discussed in Section \S\ref{sec:table}. The conducted calculations leading to said table follow analogously to Example \ref{ex:bianchi2Computation}.

%% file: conclusion.tex
\section{Summary and conclusion}
According to the standard cosmological worldview, which is based on the FLRW equations, on the largest scales the universe should be perfectly homogeneous and isotropic, adhering to the cosmological principle. However, in physical reality the universe contains a rich infrastructure of imperfections, and it is these imperfections that are essential for the formation and evolution of all structure and objects in the cosmos. These imperfections are treated in a perturbative manner, with respect to the FLRW background. Recent observations, however, further increase existing tensions within this approach, prompting explorations for alternatives.

One alternative is to change the starting point: assume that the background cosmology is not the homogeneous, isotropic FLRW. That is, that we consider cosmologies which, even at the largest scales, do not obey the cosmological principle. Perhaps such a universe (one out of many) provides a closer match to the (recent) observations.

The topic of this work was to consider homogeneous yet anisotropic universes. Seeing as homogeneity and (an)isotropy are defined in terms of the (infinitesimal) symmetries of a spacetime metric, we took this structure --- the spacetime metric's Lie algebra of Killing vector fields (KLA) --- as the foundational object. Subsequently, we turned the question around: when given a desired Lie algebra of vector fields (termed a pre-KLA), how can we find a metric on which they are Killing?

Let $\{\tensor{\xi}{_I}\}_{I=1}^3$ indicate the pre-KLA at hand, which is defined on a Riemannian manifold $\Sigma^3$, with the understanding that $\Sigma^3$ represents a spatial slice (in a spacetime decomposable topologically as $\mathcal{M}=\Sigma^3\times\mathbb{R}$). We deduced that if we can find vector fields $\{\tensor{X}{_i}\}_{i=1}^3$ (on $\Sigma^3$) such that $[\tensor{\xi}{_I},\tensor{X}{_i}]=0$ for all $i,I$, then we could construct a metric on which $\{\tensor{\xi}{_I}\}$ would be Killing: see \eqref{eq:basisInsight} in Section \S\ref{sec:findingFrame}.

By means of the method of characteristics we obtained differential equations for $\tensor{X}{_i^I}$, which are the components of $\{\tensor{X}{_i}\}$ in the $\{\tensor{\xi}{_I}\}$-basis. With the additional results of Section \S\ref{sec:auxillaries}, then, we tabulated our findings for $\tensor{X}{_i^I}$ per Bianchi type, and presented this along with the method of how to use it in Section \S\ref{sec:table}. Some example uses followed in Section \S\ref{sec:corollaries}.

The relevance of the presented formalism lies in the fact that it allows one to express a spacetime metric --- or metric-derived quantities such as the Einstein tensor --- purely based on its (desired) KVFs. Tying these directly to the KVFs in this fashion may give a more profound insight in to the dynamics of these quantities, and their relation to the underlying isometry symmetries of the metric.

This is also what we are pursuing in upcoming works of ours. We are using this more general point of view in order to investigate properties of general homogeneous, yet anisotropic universes. Our particular aim is to study cosmic microwave background (CMB) realizations in such universes, to find if there are signatures of the underlying anisotropic cosmology to be found in them. An example, discussed further in \cite{scholtens_no2}, is found in Figure \ref{fig:here}
\begin{figure}[t]
    \centering
    \includegraphics[width = \textwidth]{cmb_disc_spec.png}
    \caption{Simulation of a CMB (Sachs-Wolfe effect) in a Bianchi V universe: a superposition of wave modes permissible in this universe. The ``washing out'' towards the bottom is predicted from the form of the Bianchi V metric.}
    \label{fig:here}
\end{figure}

Another area of improvement may be to investigate the relevance of the reported formalism to locally rotationally symmetric (LRS) Bianchi models. While adding a fourth KVF into the mix a priori invalidates the utilized methodology (as it was based on a 3D Lie (pre-)KLA), the fact that such a new KVF would be rotation-generating may carry enough additional structure to incorporate it.



\section*{Acknowledgements}
The authors wish to thank Sigbjørn Hervik for a careful read, targeted feedback about the definition of a Bianchi model, and additional mathematical context. R.W.S.\ wishes to thank Dave Verweg for coining the term ``separated metric.''

%% file: references.tex
\bibliographystyle{JHEP.bst}
\bibliography{lib}